\documentclass[11pt,onecolumn]{IEEEtran} %

\usepackage[utf8]{inputenc} 
\usepackage[T1]{fontenc}
\usepackage{url}              %
\usepackage{cite}             %
\usepackage{pgfplots}
\pgfplotsset{compat=newest}

\usepackage[cmex10]{amsmath}  %
\usepackage{mleftright}       %
\mleftright                   %

\usepackage{graphicx}         %
\usepackage{booktabs}         %

\usepackage[ruled]{algorithm}
\usepackage{colortbl}
\usepackage{booktabs}
\usepackage{multirow}
\usepackage{caption}
\usepackage{float}
\usepackage{multicol}
\usepackage[normalem]{ulem}
\usepackage{algpseudocode}
\algnewcommand{\Initialize}{%
  \State \textbf{Initialize:}
}
\algnewcommand{\Output}{%
  \State \textbf{Output:}
}
\usepackage{mathtools}
\usepackage{graphicx}
\usepackage{url}

\usepackage[colorlinks=true, allcolors=blue]{hyperref}
\usepackage[
    shortcuts,       %
    acronym,         %
    section=section, %
]{glossaries}
\usepackage{tikz}
\usetikzlibrary{backgrounds,calc,shadings,shapes.arrows,shapes.symbols,shadows,fit,positioning,topaths,hobby}
\usepackage{float}
\usepackage{amssymb}
\usepackage{amsmath}
\usepackage{amsthm}
\usepackage{xspace}
\usepackage[capitalize]{cleveref}
\usepackage{subfig}
\usepackage[inline]{enumitem}
\usepackage{cite}
\usepackage{siunitx}
\usepackage{dsfont}
\usepackage{balance}
\makeglossaries

\makeatletter
\DeclareRobustCommand{\extstart}{%
  \@bsphack
  \leavevmode
  \normalcolor %
  \@esphack
}
\DeclareRobustCommand{\extend}{%
  \@bsphack
  \normalcolor
  \@esphack
}

\newtheorem{theorem}{Theorem}

\newtheorem{lemma}{Lemma}

\newtheorem{proposition}{Proposition}
\newtheorem{definition}{Definition}
\newtheorem{claim}{Claim}

\newtheorem{assumption}{Assumption}

\newtheorem{remark}{Remark}

\IEEEoverridecommandlockouts

\definecolor{bleudefrance}{rgb}{0.19, 0.55, 0.91}

\newcommand{\cF}{\mathcal{F}}

\newcommand{\cH}{\mathcal{H}}

\newcommand{\cO}{\mathcal{O}}
\newcommand{\cP}{\mathcal{P}}

\newcommand{\cX}{\mathcal{X}}
\newcommand{\cY}{\mathcal{Y}}
\newcommand{\cZ}{\mathcal{Z}}

\newcommand{\rvX}{\mathsf{X}}
\newcommand{\rvY}{\mathsf{Y}}
\newcommand{\rvZ}{\mathsf{Z}}

\newcommand{\inpara}[1]{\ensuremath{\left(#1\right)}}

\newcommand{\inbracelets}[1]{\ensuremath{\left\{#1\right\}}}

\newcommand{\EE}{\mathbb{E}}

\newcommand{\entropy}[2]{\ensuremath{\mathrm{H}_{#1}\left( #2 \right)}}

\newcommand{\mutinf}[2]{\ensuremath{\mathrm{I}_{#1}\left( #2 \right)}}

\newcommand{\kl}[2] {\ensuremath{\mathrm{D}_q \left( #1 \| #2 \right)}}

\definecolor{my_blue}{rgb}{0,0,0.8}
\definecolor{my_darkblue}{rgb}{0.2,0.2,0.6}
\definecolor{my_lightblue}{rgb}{0,0,0.8}
\definecolor{shadecolor}{gray}{.65} 
\definecolor{my_red}{rgb}{0.63922,0.14902,0.21961}
\definecolor{my_redgray}{rgb}{0.5,0.14902,0.21961}
\definecolor{my_lightred}{rgb}{0.63922,0.14902,0.21961}
\definecolor{my_green}{rgb}{0.2,0.6,0.25}
\definecolor{my_lightgreen}{rgb}{0.5,0.9,0.55}
\definecolor{blond}{rgb}{0.98, 0.94, 0.75}
\definecolor{caucasian}{rgb}{1,0.8,0.6}
\definecolor{richelectricblue}{rgb}{0.03, 0.57, 0.82}
\definecolor{bleudefrance}{rgb}{0.19, 0.55, 0.91}
\definecolor{ogreen}{rgb}{0,0.5,0}
\definecolor{xgray}{rgb}{0.9, 0.9, 0.9}

\definecolor{TUMBlue}{RGB}{0,101,189} %
\definecolor{TUMBlueDark}{RGB}{0,82,147} %
\definecolor{TUMBlueLight}{RGB}{152,198,234} %
\definecolor{TUMBlueMedium}{RGB}{100,160,200} %

\definecolor{TUMIvory}{RGB}{218,215,203} %
\definecolor{TUMGreen}{RGB}{162,173,0} %
\definecolor{TUMGray}{gray}{0.6} %

\definecolor{TUMGreenDark}{RGB}{0,124,48} %
\definecolor{TUMRed}{RGB}{196,7,27} %

\newcommand{\define}{\ensuremath{:=}}
\DeclareMathOperator*{\argmax}{arg\,max}

\newcommand{\timeidx}{\ensuremath{n}}
\newcommand{\nchannels}{\ensuremath{k}}
\newcommand{\symtx}{\ensuremath{\mathcal{X}}}
\newcommand{\symrx}{\ensuremath{\mathcal{Y}}}

\newcommand{\channelidx}{\ensuremath{j}}
\newcommand{\channel}[1][\channelidx]{\ensuremath{W_{#1}}}
\newcommand{\channellemma}{\ensuremath{V}}

\renewcommand{\mutinf}[1]{\mathrm{I}\left(#1\right)}
\renewcommand{\entropy}[1]{\mathrm{H}\left(#1\right)}
\newcommand{\nsamples}{\ensuremath{n}}
\newcommand{\confidence}{\ensuremath{\delta}}
\newcommand{\roundidx}{\ensuremath{r}}
\newcommand{\roundchannels}[1][\roundidx]{\ensuremath{\mathcal{C}_{#1}}}
\newcommand{\epsci}{\ensuremath{\varepsilon}}
\newcommand{\confci}{\ensuremath{\alpha}}
\newcommand{\epsround}[1][\roundidx]{\ensuremath{\varepsilon_{#1}}}
\newcommand{\deltaround}[1][\roundidx]{\ensuremath{\confidence_{#1}}}
\newcommand{\pullsround}[1][\roundidx]{\ensuremath{\timeidx_{#1}}}
\newcommand{\pacchannelround}{\ensuremath{\channelidx_\roundidx}}
\newcommand{\suboptimalitygap}[1][\channelidx]{\Delta_{#1}}
\newcommand{\capacitychannel}[1][\channelidx]{\ensuremath{\mathrm{C}\left(W_{#1}\right)}}
\newcommand{\capacitychannelest}[1][\channelidx]{\ensuremath{\hat{\mathrm{C}}^{n_{#1}}(W_{#1})}}

\newcommand{\capacitychannelestn}[1][\channelidx]{\ensuremath{\hat{\mathrm{C}}^{n}(W_{#1})}}
\newcommand{\capacitychannelestr}[1][\channelidx]{\ensuremath{\hat{\mathrm{C}}^{\pullsround}(W_{#1})}}
\newcommand{\history}{\cH}
\newcommand{\outputchannel}{\hat{\channelidx^\star}}
\newcommand{\pacoutputchannel}{\hat{\channelidx}_{\epsci}}
\newcommand{\ourpolicy}{\texttt{BestChanID}\xspace}
\newcommand{\pacpolicy}{\texttt{NaiveChanSel}\xspace}
\newcommand{\medianpac}{\texttt{MedianChanEl}\xspace}
\newcommand{\timeroundstart}[1][\roundidx]{\ensuremath{\tau_{#1}}\xspace}
\newcommand{\maxrounds}{\ensuremath{R}}

\newcommand{\channelpart}{\ensuremath{\mathcal{A}_s}}
\newcommand{\channelpartround}[1][\roundidx]{\ensuremath{\mathcal{C}_{#1, s}}}

\newcommand{\card}[1]{\ensuremath{\vert #1 \vert}}

\renewcommand{\kl}[2]{\ensuremath{\mathrm{D}\left(#1 \|#2\right)}}

\newcommand{\bentropy}[1]{\ensuremath{h_\textrm{b}\left(#1\right)}}

\newcommand{\dtv}[2]{\ensuremath{d_{\textrm{TV}}\left(#1,#2\right)}}

\newcommand{\dist}{\ensuremath{w}}
\newcommand{\change}{\ensuremath{a}}
\newcommand{\prmeasure}{\ensuremath{\mathcal{W}}}
\newcommand{\kldiff}{\ensuremath{\zeta}}
\newcommand{\sampleidx}{\ensuremath{l}}
\newcommand{\nsenses}[1][\channelidx]{\ensuremath{N_{x, #1}}}
\newcommand{\obssenses}{\ensuremath{Z_{x, \channelidx, l}}}

\newcommand{\linfac}{\ensuremath{c_1}}
\newcommand{\logfac}{\ensuremath{c_2}}
\newcommand{\tmpcrs}{\ensuremath{c_3}}

\newcommand{\logfactor}{\ensuremath{\zeta}} %
\newcommand{\tmpx}{\ensuremath{\nu}}
\newcommand{\logrcubed}[1][\roundidx]{\ensuremath{f(#1)}}
\newcommand{\rootfactor}{\ensuremath{\gamma}} %
\newcommand{\rootfactorsplit}{\ensuremath{2\rootfactor}} %
\newcommand{\nfactor}[1][\alpha]{\ensuremath{g\left(#1\right)}}
\newcommand{\nfactorprime}[1][\deltaround]{\ensuremath{\Tilde{g}(\alpha)}}%
\newcommand{\logfacinst}[1][\alpha]{\ensuremath{t(#1)}}
\newcommand{\logfacvar}[1][\alpha]{\ensuremath{\Tilde{t}(#1)}}
\newcommand{\tmpconst}{\ensuremath{\rho}}
\newcommand{\logsq}{\ensuremath{b(\deltaround, \epsround})}

\begin{document}

\title{Maximal-Capacity Discrete Memoryless\\ Channel Identification} %

\author{%
    \IEEEauthorblockN{Maximilian Egger, %
                    Rawad Bitar, %
                    Antonia Wachter-Zeh, %
                    Deniz Gündüz and %
                    Nir Weinberger%
                    } \vspace{-.5cm}

    \thanks{M.E., R.B. and A.W-Z. are with the Technical University of Munich. Emails: \{maximilian.egger, rawad.bitar, antonia.wachter-zeh\}@tum.de. D.G. is with Imperial College London. Email: d.gunduz@imperial.ac.uk. N.W. is at Technion --- Israel Institute of Technology. Email: nirwein@technion.ac.il.}
    \thanks{This project has received funding from the German Research Foundation (DFG) under Grant Agreement Nos. BI 2492/1-1 and WA 3907/7-1. The work of N.W. was partly supported by the Israel Science Foundation (ISF), grant no. 1782/22. \extstart Parts of the results were presented at IEEE International Symposium on Information Theory (ISIT), 2023 \cite{egger2023maximal}. \extend}}

\maketitle

\begin{abstract}
  The problem of identifying the channel with the highest capacity among several discrete memoryless channels (DMCs) is considered. The problem is cast as a pure-exploration multi-armed bandit problem, which follows the practical use of training sequences to sense
  the communication channel statistics. A capacity estimator is proposed and tight confidence bounds on the estimator error are derived. Based on this capacity estimator, a gap-elimination algorithm termed \ourpolicy is proposed, which is oblivious to the capacity-achieving input distribution and is guaranteed to output the DMC with the largest capacity, with a desired confidence. \extstart Furthermore, two additional algorithms \pacpolicy and \medianpac, that output with certain confidence a DMC with capacity close to the maximal, are introduced. Each of those algorithms is beneficial in a different regime \extend and can be used as a subroutine in \ourpolicy. The sample complexity of \extstart all \extend algorithms %
  is analyzed as a function of the desired confidence parameter, the number of channels, and the channels' input and output alphabet sizes. The cost of best channel identification is shown to scale \extstart quadratically with the alphabet size, and a fundamental lower bound for the required number of channel senses to identify the best channel with a certain confidence is derived. \extend
\end{abstract}

\begin{IEEEkeywords}
Best-arm identification, Capacity estimation, Channel identification, Discrete Memoryless Channels, Multi-Armed Bandits, Pure-exploration, Sample complexity
\end{IEEEkeywords}

\section{Introduction}

We consider an information-theoretic instance of a pure exploration problem in the setting of noisy channel coding. From a communication system perspective, exploration \extstart of the channel state \extend is an important task, which must be routinely performed in order to maintain the highest possible communication rate based on the current channel conditions \cite[Ch. 10]{proakis2001digital}\cite[Ch. 9]{barry2012digital}. Typically, in such systems, before actual data is being transmitted, or whenever an uncertainty of the channel conditions arises, the transmitter sends a training sequence that is agreed upon in advance. Based on the received signal and the knowledge of the training sequence, the receiver estimates certain statistical characteristics of the channel, and then may adapt its various components accordingly, such as equalizers and decoders. In many cases, the receiver can also communicate its acquired knowledge back to the transmitter (e.g., in a two-way link). In this case, the transmitter can also adapt its operation, e.g., its encoding rate, its codebook, physical parameters of the modulator, and so on. %
The duration of the training phase is a natural performance criterion in wireless systems, along with the quality of the obtained learned components and the reliability of the decision. 

A prototypical and simple instance of this general channel estimation/learning problem is that of channel \textit{selection}. Consider a transmitter that can communicate to the receiver using one of $\nchannels$ parallel channels, whose statistical properties are unknown in advance. In this scenario, the adaptation problem becomes a simple decision problem between $\nchannels$ options. For instance, orthogonal-frequency-division-multiplexing (OFDM) systems use $\nchannels$ different frequency bands. The system may choose which of the available frequency bands to allocate to a user, typically those with the strongest frequency response and hence, the largest capacity \cite{liu2014channel}. After choosing the active sub-channels, the system may optimize its power allocation via the water-pouring rule \cite[Ch. 9]{cover2012elements}, or even both rate- and power-allocation \cite[Ch. 9]{goldsmith2005wireless}. Similarly, in multiple-input multiple-output (MIMO) systems \cite{tse2005fundamentals,biguesh2006training}, the receiver may select the transmit-antenna with the maximal fading gain.  

In the case of Gaussian channels, the objective reduces to parameter estimation, i.e., estimating the channel gain or the channel matrix in the case of MIMO; since the same Gaussian input distribution maximizes the capacity of all the channels. Instead, we consider this problem in the case of discrete memoryless channels (DMCs). To the best of our knowledge, and to some extent, to our surprise, this problem has not been studied in the literature. 
Specifically, we assume that communication may take place over one of $\nchannels$ possible DMCs with the same finite input and output alphabets. During the training phase, the transmitter may transmit training symbols over each channel and sample their output. Based on these samples, the receiver's goal is to decide which channels have the maximal capacity, which can be achieved upon optimizing the input distribution. From the perspective of multi-armed bandit (MAB) problems, this falls into the category of \textit{pure-exploration}  \cite[Ch. 33]{lattimore2020bandit}. In the standard MAB problem, the agent sequentially chooses (pulls) one of $\nchannels$ possible arms to maximize its cumulative reward. To achieve this without knowing the reward of each arm in advance, the agent has to balance exploration and exploitation. In contrast, in the pure-exploration problem of best-arm identification in the fixed-confidence setting \cite{jamieson2014best}, the only goal is to find the arm with maximal mean reward, with a desired reliability. In contrast to the standard MAB problem, the potential loss during exploration is not accumulated, and the agent's decision only has to comply with the target reliability level while exploiting a minimal number of training symbols.
\extstart
The fundamental difficulty of such a problem is quantified by lower bounds that provide insights on the minimum number of required arm pulls to determine the best arm with a certain confidence, e.g., \cite{audibert2010best,kaufmann2016complexity,garivier2016optimal}.

One variant of the best-arm identification problem relaxes the objective from outputting the maximal-reward arm to only outputting an arm whose reward is $\varepsilon$-close to the maximal expected reward under a given target reliability level, say, with probability at least $1-\delta$. This naturally reduces the number of rounds required for a given reliability. Such algorithms are termed $(\varepsilon, \delta)$-PAC, abbreviated for probably approximately correct (PAC) \cite{evendar2006action}.
\extend

Exploration in wireless communication networks has been formulated as a MAB problem
by various authors, e.g., \cite{coucheney2015multi,lunden2015spectrum,blasco2015multi,oksanen2016machine,maghsudi2016multi,nikfar2017relay,yang2018multi,boldrini2018mumab,besson2018aggregation,bonnefoi2018multi}. The problem of estimating Shannon-theoretic measures is well-explored from various perspectives, starting from \cite{Antos2001,Paninski2003}, and then extended to large alphabets \cite{Jiao2015,Jiao2017,wu2016minimax},
to the rate-distortion function \cite{harrison2008estimation}, to analog
sources \cite{wang2009universal}\extstart, to mutual information \cite{shamir2010learning,stefani2014confidence} \extend and more. A MAB problem, in which
the instantaneous reward is based on self-information, was recently
considered in \cite{weinberger2022upper,weinberger2022multi}. A DMC input distribution adaptation method was proposed and analyzed in~\cite{tridenski2019channel}. \extstart The problem of best channel identification through MABs with pure exploration was first introduced in an earlier version of this paper \cite{egger2023maximal}.
\extend

\extstart
\paragraph{Contributions}
We formulate the problem of identifying the DMC with the maximal capacity among a set of $\nchannels$ DMCs as a MAB problem in the pure-exploration regime with fixed confidence. We propose a method to estimate the capacity of a DMC by sampling its output with a given sequence of known input symbols. We provide confidence intervals for the error of the estimator, of a sub-Gaussian nature, that is, an interval length whose size is roughly $O(\sqrt{{\log({1}/{\confidence})}/{\nsamples}})$, for $\nsamples$ samples and confidence level $\confidence$. To obtain this confidence interval for the capacity estimation error, we invoke the minimax characterization of capacity \cite{csiszar1972class,kemperman1974shannon}, which expresses the capacity as the minimization over output distributions, cf.~\eqref{eq: Csiszar minimax capacity}, rather than maximization over input distributions. This allows us to define a \textit{pseudo-capacity}, which replaces the minimization of the output alphabet to a strict interior of the simplex; concretely, the probability of each output letter is restricted to be larger than $\eta >0$. The resulting pseudo-capacity approximates the true capacity, and the value of $\eta$ controls the bias-variance trade-off in the estimation error analysis. We utilize the confidence interval bound on the capacity estimation error and propose a gap-elimination algorithm for identifying the DMC with the maximal capacity from $\nchannels$ possible candidates, with high probability. We introduce two PAC algorithms that are optimal in different parameter regimes and can be used as standalone or in conjunction with the proposed gap-elimination algorithm. These algorithms identify an $\epsci$-best DMC whose capacity is, at a high level, at most $\epsci$ far from that of the best channel. We analyze the sample complexity of the proposed algorithms both theoretically and experimentally. Given a set of candidate DMCs parameterized by their transition matrices, we derive a fundamental lower bound on the number of channel senses required to identify the one with the highest capacity with a probability of at least $1-\delta$.

\paragraph{Outline}
In \cref{sec:prelim}, we introduce the system model for exploring DMCs with unknown capacities through sensing. In \cref{sec:confidence}, we establish a confidence bound for capacity estimation under deterministic sensing. In \cref{sec:best_channel_id}, we introduce \ourpolicy, a round-based elimination algorithm for best channel identification in the context of MABs with fixed confidence, which relies on the existence of algorithms that output a channel with capacity being close to the capacity of the best channel. In \cref{sec:pac}, we investigate such algorithms. Lastly,  in \cref{sec:converse}, we establish a fundamental converse result of a lower bound on the number of required channel senses for best channel identification. We conclude the paper in \cref{sec:conclusion}. \extstart For a smooth flow of the paper, the majority of the proofs are deferred to the appendix. \extend
\extend

\section{Preliminaries and System Model}\label{sec:prelim}
Throughout the paper, random variables are denoted by typesetter letters, e.g., $\rvZ$, and sets by calligraphic letters, e.g., $\cZ$. The probability distribution of a random variable $\rvX \in \cX$ is denoted by $P_\rvX$, and the probability simplex over the alphabet $\cX$ is denoted by $\cP(\cX)$. For integers $r_1,r_2\in \mathbb{N}_{+}$ such that $r_1<r_2$,  $[r_1] \define \{1,\dots,r_1\}$ and $[r_1,r_2] \define \{r_1,\dots,r_2\}$. All logarithms are taken to the natural base unless stated otherwise. The Kullback-Leibler (KL) divergence between two distributions $P_\rvX$ and $Q_\rvX$ is denoted by $\kl{P_\rvX}{Q_\rvX}$, the binary entropy function by $\bentropy{t}:=-t\log(t)-(1-t)\log(1-t)$, and the total variation distance by $\dtv{P_\rvX}{Q_\rvX}:=\sum_{x\in\mathcal{X}}|P_\rvX(x)-Q_\rvX(x)|$. Let alphabets $\cX$ and ${\cY}$ be given. The mutual information functional between the input distribution $P_{\rvX}$ over $\cX$ and the output distribution over ${\cY}$ induced by channel $W_\channelidx$ is denoted by $\mutinf{P_{\rvX};W_\channelidx}$, and the capacity of the DMC $W_{\channelidx}$ is denoted by $\capacitychannel$. 

Let $\{W_{\channelidx}\}_{\channelidx\in[k]}$
be a collection of $k$ DMCs, $W_{\channelidx}\colon{\cX}\to{\cY}$. 
The goal of the \emph{learner} is to identify the index  $\channelidx^\star$ of the DMC with the largest capacity, i.e., $\channelidx^{\star}\define\argmax_{\channelidx \in [k]} \capacitychannel$, which we assume to be unique (otherwise, identifying
the best channel is impossible), and that the maximal channel capacity is finite.
The \emph{suboptimality gap} $\suboptimalitygap$ of a certain channel $\channelidx$ is defined as the difference between its capacity and the capacity of the best channel, i.e.,
\begin{align*}
    \suboptimalitygap \define \capacitychannel[\channelidx^{\star}] - \capacitychannel.
\end{align*}

To identify the best channel, the learner is given limited access to sense the DMCs, such that for each channel \extstart $\channelidx$ she can obtain a set of input-output samples $\{X_{i},Y_{i}\}$, \extend where $Y_{i}\sim W_\channelidx(\cdot\mid X_{i})$ are independently and identically distributed (i.i.d).

Specifically, at time $t\in\mathbb{N}_{+}$, the learner chooses
a symbol $X_{t}\in{\cX}$ and a channel index $\channelidx_{t}$ and observes
the output $Y_{t}\sim W_{\channelidx_{t}}(\cdot\mid X_{t})$. The action of the
learner is thus the pair $A_{t}=(X_{t},\channelidx_{t})$, and is a function
of the past actions and observations $\history_{t}:=\{(X_{s},\channelidx_{s},Y_{s})\}_{s\in[t-1]}$
(and thus is random). A policy $\pi=(\pi_{t})_{t\in\mathbb{N}_{+}}$
is a mapping from the history $\history_{t}$ to an action, $A_{t}=\pi_{t}(\history_{t})$.

Given a fixed confidence level $\delta>0$, the learner is required to find a policy $\pi$ 
and an associated stopping time $\tau$ that is adapted to the filtration ${\cF}=({\cF}_{t})_{t\in\mathbb{N}_{+}}$, where ${\cF}_{t}=\sigma(\history_{t})$. The policy determines which channel the learner should sample at each time instance $t\in [\tau]$ before stopping. In addition, the learner is required to find an ${\cF}_{\tau}$-measurable selection function $\psi$ for which \mbox{$\Pr\left[\tau<\infty\text{ and }\psi(\history_{\tau})\neq \channelidx^{\star}\right]\leq\delta$},
that also achieves the minimal $\mathbb{E}[\tau]$. In other words, the figure of merit of a policy $\pi$ is its expected stopping time. \extstart Assume that the learner has decided to stop (because $t=\tau(\history_{t})$) such that for each channel she is given $\nsamples_\channelidx$ input-output samples $\{X_{i},Y_{i}\}_{i\in[\nsamples_\channelidx]}$. Since the channels are independent, \extend
the optimal decision function $\phi$
uses the samples from each respective channel to obtain an estimate $\capacitychannelest$ of its capacity $\capacitychannel$, and outputs the channel with the highest estimate.

\section{Capacity Estimator and A Confidence Interval}\label{sec:confidence}

In this section, we propose an estimator for the capacity of a DMC
$W\colon{\cal X}\to{\cal Y}$, given by $\mathrm{C}(W)=\max_{P_{\rvX}}\mutinf{P_{\rvX};W}$, and state a confidence interval on the estimation error. We omit the channel index $\channelidx$ in this section for ease of notation.
To estimate the capacity, the learner is provided with $\nsamples$
samples of the channel output obtained for her choice of inputs. A
straightforward approach to capacity estimation is to decompose the
mutual information $\mutinf{P_{\rvX};W}$ as $\mutinf{P_{\rvX};W}=\entropy{P_{\rvY}}-\entropy{P_{\rvY|\rvX}\mid P_{\rvX}}$,
and then separately estimate each of the entropy terms. Confidence
intervals for the entropy terms will then lead to a confidence interval
for the mutual information for a fixed $P_{\rvX}$. Such confidence intervals
for entropy estimation were originally derived in \cite{Antos2001}
(see also \cite{shamir2010learning,Ho_Yeung2010,weinberger2022upper}).
Nonetheless, obtaining a confidence interval for the \emph{capacity} of a DMC is more challenging since the capacity is the maximum
of the mutual information over \textit{all} possible input distributions. To
address this issue, we utilize the \emph{minimax characterization of capacity \cite{csiszar1972class,kemperman1974shannon}} to obtain uniform convergence bounds. This characterization will also justify our choice to allocate an equal number of samples to
each input letter in ${\cal X}$. 

In this section, it will be convenient to refer to $W$ as a collection of conditional distributions
$\{W_{\rvY|\rvX=x}\}_{x\in{\cal X}}$. With
this notation, we may recall that the minimax capacity
theorem \cite{csiszar1972class,kemperman1974shannon} states that 
\begin{equation}
\mathrm{C}(W)=\min_{Q_{\rvY}\in{\cal P}({\cal Y})}\max_{x\in{\cal X}}\kl{W_{\rvY|\rvX=x}}{ Q_{\rvY}},\label{eq: Csiszar minimax capacity}
\end{equation}
where ${\cal P}({\cal Y})$ is the probability simplex over ${\cal Y}$.
A simple consequence of the proof of~\eqref{eq: Csiszar minimax capacity}
is that the minimizer $Q_{\rvY}^\star$ is in fact the capacity-achieving output
distribution $P_{\rvY}^{\star}$, i.e., the output distribution induced by
a capacity achieving input distribution $P_{\rvX}^{\star}$. With a slight
abuse of notation, let us denote the strict interior of ${\cal P}({\cal Y})$,
specifically, the set of output distributions for which any symbol
has a probability mass larger than $\eta>0$, as
\[
{\cal P}_{\eta}({\cal Y}):=\{Q_{\rvY}\in{\cal P}({\cal Y})\colon\min_{y\in{\cal Y}}Q_{\rvY}(y)\geq\eta\}.
\]
We may then define the \emph{pseudo-capacity }
\begin{equation}
\mathrm{C}_{\eta}(W):=\min_{Q_{\rvY}\in{\cal P}_{\eta}({\cal Y})}\max_{x\in{\cal X}}\kl{W_{\rvY|\rvX=x}}{Q_{\rvY}},\label{eq: pseudo-capacity}
\end{equation}
which is obtained by replacing the minimization over ${\cal P}({\cal Y})$
with a minimization over ${\cal P}_{\eta}({\cal Y})\subset{\cal P}({\cal Y})$.
It should be noted that the solution $Q_{\rvY}^{\star}$ of the minimization
problem in~\eqref{eq: pseudo-capacity} may not lead to a valid output
distribution, in the sense that there may be no input
distribution $P_{\rvX}$ for which $\sum_{x\in{\cal X}}P_{\rvX}(x)W_{\rvY|\rvX}(y | x)$
corresponds to $Q_{\rvY}^{\star}$. For example, the erasure symbol for a
binary erasure channel with erasure probability $\eta/2$ cannot have
probability larger than $\eta$, no matter what the input distribution
is. Nonetheless, the value of $\mathrm{C}_{\eta}(W)$ will serve as
an upper bound to the true capacity. 

By continuity of the KL divergence in the interior of the simplex,
it can be noted that $\mathrm{C}(W)=\lim_{\eta\downarrow0}\mathrm{C}_{\eta}(W)$.
The following lemma is a refinement of this observation: 
\begin{lemma}
\label{lem: capacity bias}For any DMC $\channellemma \colon{\cal X}\to{\cal Y}$, it holds that
\[
0\leq\mathrm{C}_{\eta}(\channellemma)-\mathrm{C}(\channellemma)\leq \extstart 2\eta |{\cal Y}|.\extend
\]
\end{lemma}
\begin{proof}

\extstart
To prove \cref{lem: capacity bias}, we rely on the following result that bounds from above the minimum KL divergence between a reference distribution $P_Y \in {\cal P(Y)}$ and any $Q_{Y}^*\in{{\cal P}_{\eta}({\cal Y})}$.
\begin{lemma} \label{lemma:kl_bound} Let $\eta\in\left(0,1/(2\card{\symrx})\right)$. For any $P_{Y}\in{\cal P}({\cal Y})$, we have
    $$\min_{Q_{Y}\in{{\cal P}_{\eta}({\cal Y})}} \kl{P_{Y}}{Q_{Y}} %
    \leq 2 \eta \card{\symrx}.$$
\end{lemma}
\begin{proof}
    The proof is provided in \cref{app:proof_kl_bound}.
\end{proof}
We now prove \cref{lem: capacity bias}. It is obvious that $\mathrm{C}_{\eta}(\channellemma)\geq\mathrm{C}(\channellemma)$, since ${\cal P}_{\eta}({\cal Y})\subset{\cal P}({\cal Y})$.
Furthermore, 
\begin{align*}
\mathrm{C}_{\eta}(\channellemma) & =\min_{Q_{Y}\in{\cal P}_{\eta}({\cal Y})}\max_{x\in{\cal X}}D(V_{Y|X=x}\mid\mid Q_{Y})\\
 & =\min_{Q_{Y}\in{\cal P}_{\eta}({\cal Y})}\max_{P_{X}\in{\cal P}({\cal X})}D(\channellemma_{Y|X}\mid\mid Q_{Y}\mid P_{X})\\
 & \overset{{\scriptstyle (a)}}{=}\max_{P_{X}\in{\cal P}({\cal X})}\min_{Q_{Y}\in{\cal P}_{\eta}({\cal Y})}D(\channellemma_{Y|X}\mid\mid Q_{Y}\mid P_{X})\\
 & \overset{{\scriptstyle (b)}}{\leq}\max_{P_{X}\in{\cal P}({\cal X})} \left[I(P_{X}; \channellemma)+\min_{Q_{Y}\in{\cal P}_{\eta}({\cal Y})}D(P_{Y}\mid\mid Q_{Y})\right] \\
 & =\mathrm{C}(\channellemma)+\max_{P_{X}\in{\cal P}({\cal X})} \min_{Q_{Y}\in{\cal P}_{\eta}({\cal Y})}D(P_{Y}\mid\mid Q_{Y}) \\
 &\overset{\scriptstyle (c)}{\leq} \mathrm{C}(\channellemma) + 2 \eta \card{\symrx},
\end{align*}
where $(a)$ follows from the minimax theorem (${\cal P}({\cal X})$
and ${\cal P}_{\eta}({\cal Y})$ are convex sets, and the average
KL divergence is linear in $P_{X}$, hence concave, and convex in
$Q_{Y}$); in $(b)$ $P_{Y}$ is the marginal induced by the input
$P_{X}$ and $\channellemma$, $P_{Y}(y)=\sum_{x\in{\cal X}}P_{X}(x)\channellemma_{Y|X}(y\mid x)$; and lastly, $(c)$ follows since \cref{lemma:kl_bound} holds uniformly for any $P_{Y}$.
\extend
\end{proof}

Consider the following procedure for estimating the capacity of $W$, given a total budget of $\nsamples$ channel sensing operations. 
Each input symbol $x\in{\cal X}$ is fed into the channel $\frac{\nsamples}{|{\cal X}|}$
times (ignoring for simplicity the integer constraints as they do not
substantially affect the algorithm or its analysis), and the empirical conditional distribution
of $\rvY$ denoted by $\hat{W}_{\rvY|\rvX=x}$ is computed as $\hat{W}_{\rvY|\rvX}(y|x)=\frac{|{\cal X}|}{\nsamples}\sum_{i\in[\nsamples/|{\cal X}|]}\mathds{1}\inbracelets{{Y_{i}=y\mid \rvX=x}}$, where ${Y_{i}}$, $i\in[\nsamples/|{\cal X}|]$, are the output samples corresponding to input $x\in\cal X$. The capacity is then estimated
in a natural way as 
\[
\hat{\mathrm{C}}^{\nsamples}(W)\define\mathrm{C}(\hat{W})=\max_{P_{\rvX}\in {\cal P}({\cal X})}\mutinf{P_{\rvX};\hat{W}},
\]
and the chosen input distribution is any $\hat{P}_{\rvX}^{\star}\in\argmax_{P_{\rvX} \in {\cal P}({\cal X})} \mutinf{P_{\rvX};\hat{W}}$. 
\extstart 
The uniform allocation of samples to input symbols is justified by the minimax formulation of capacity, which suggests to target similar worst-case estimation errors for every row of the transition matrix (for every $x\in \symtx$). Works such as \cite{shulman2004uniform} further show that under some conditions the uniform distribution is a reasonable prior for every channel. However, note that we only use a uniform allocation of the input symbols during the sensing period and optimize over the input distribution for determining the capacity of a channel. Otherwise, we would suffer a non-negligible additive loss in the capacity calculation for each channel.
\extend
For clarity of exposition, we define \extstart 
\begin{align*}
    \nfactor &\define 4 \card{\symtx} \card{\symrx} \log(\card{\symtx}/\alpha) \text{ and } \\
    \logfacinst &\define \sqrt[5]{\frac{\card{\symrx} e^2}{\card{\symtx}\log\frac{|{\cal X}|}{\alpha}}}.
\end{align*}
Given those two functions, we state in \cref{prop: Concentration for capacity} the confidence bound for capacity estimation.
\begin{proposition}
\label{prop: Concentration for capacity} Let $\alpha \leq 1$ be given, then for 
\begin{align*}
    \varepsilon=\frac{5\sqrt{\nfactor}}{4} \frac{\log(\nsamples \cdot \logfacinst)}{\sqrt{\nsamples}}+\frac{\nfactor}{\nsamples},
\end{align*} \extend
we have
\[
\Pr\left[\left|\hat{\mathrm{C}}^\nsamples(W) -\mathrm{C}(W)\right|\leq\varepsilon\right]\geq1-\alpha.
\]
\extstart
Assuming $\card{\symtx} = \card{\symrx} \geq 2$ and $0 \leq \alpha \leq 1/2$ so that
$\log\frac{|{\cal X}|}{\alpha}>1$, the statement for $\varepsilon$ simplifies by bounding $\logfacinst$ by a small constant $\logfacinst \leq 1.1$.
\extend
\end{proposition}
We provide a sketch of the proof for brevity. The full proof is given in \cref{app:proof_concentration_capacity}.
\begin{proof}[Sketch of Proof]
    \extstart We use the dual formulation of capacity in \eqref{eq: Csiszar minimax capacity} and the triangle inequality to express the difference between $C(\hat{W})$ and $C(W)$ by the bias of the pseudo-capacities $C_\eta(W)$ and $C_\eta(\hat{W})$, which can be bounded as in \cref{lem: capacity bias}, and a probabilistic bound on the difference between $C_\eta(W)$ and $C_\eta(\hat{W})$. The latter follows from expressing the difference of capacities using the minimax characterization in \eqref{eq: Csiszar minimax capacity} decomposed into entropy and cross-entropy, which can be individually bounded from above. Setting $\eta=1/\nsamples$ provides a trade-off between the different bounds and leads to the statement above.
    \extend
\end{proof}

\extstart
\begin{remark}
    We use the dual formulation of capacity for two reasons. First, the capacity achieving output distribution is unique, and so optimizing over this distribution is more natural. By contrast, the capacity achieving input distribution may not be unique, and optimal input distributions may even have completely disjoint supports. Thus, the support may not be stable with respect to estimation errors. 
    Second, our confidence interval is obtained by approximating the capacity with pseudo-capacity, and then estimating the pseudo-capacity. The definition of the pseudo-capacity is natural in the dual formulation, because it is based on an explicit constraint on the output probability. The primal formulation of the pseudo-capacity would constrain the input distributions to ones whose output distribution satisfies that constraint, and thus is less explicit. Without any assumptions on the input or output distributions, it is unclear how to provide a tight confidence interval for capacity estimation. 

\end{remark}
\extend

Given the confidence interval bound in \cref{prop: Concentration for capacity}, we state in \cref{thm:req_samples} the number of samples $\nsamples$ required to achieve a certain confidence level $\confci$ for the capacity estimate $\hat{\mathrm{C}}^\nsamples(W)$. For ease of notation, we ignore ceiling operators as their impact vanishes throughout the asymptotic analysis.
\extstart Let $\linfac \define \frac{15\cdot 25}{4}$ be a constant and $\logfacvar \define \frac{25}{4} \logfacinst$ be a scaled version of $\logfacinst$.\extend

\begin{lemma} \label{thm:req_samples}
For $\nsamples$ satisfying~\eqref{eq:req_samples} \extstart and a given $\epsci>0$\extend, the estimated capacity has an error of at most $\epsci$ with probability at least $1-\alpha$:
\extstart
\begin{align}
    \nsamples \geq \max\inbracelets{\frac{\linfac \nfactor}{\epsci^2} \log^2\left(\frac{\logfacvar \nfactor}{\epsci^2}\right), \frac{2 \nfactor}{\epsci}}. \label{eq:req_samples}
\end{align}
For $\epsci \leq 1$, the sufficient number of samples is bounded by $n \geq \frac{15 \linfac \nfactor}{\epsci^2} \log^2\left(\frac{\logfacvar \nfactor}{\epsci^2}\right)$. \extend
\begin{proof}
    The proof is given in \cref{app:proof_req_samples}.
\end{proof}
\end{lemma}
This result will be used throughout the paper to determine the number of channel senses used by our algorithms to determine the best channels with a certain probability.

\section{Best Channel Identification}\label{sec:best_channel_id}

We now introduce a policy \ourpolicy whose output $\outputchannel$ corresponds to the channel with the highest capacity with probability at least $1-\confidence$, i.e., $\Pr(\outputchannel = \channelidx^\star)\geq 1-\confidence$. We give in \cref{thm:correctness_complexity} the stopping time of \ourpolicy reflected by the total number of channel senses and referred to as the \emph{sample complexity}. Our algorithm \ourpolicy proceeds in rounds and gradually removes channels that are unlikely to be the best channel from the set of possible future actions. In the context of MABs, this falls into the class of \textit{action elimination} algorithms. In particular, this algorithm is an adaptation of the exponential gap-elimination algorithm proposed in \cite{karnin2013almost}.

Let $\maxrounds$ be the maximum number of rounds for which \ourpolicy is run, and let $\{\roundchannels \}_{\roundidx \in[\maxrounds]}$ be a collection of nested sets with monotonically decreasing cardinality such that $\{\outputchannel\} =\roundchannels[\maxrounds] \subset \roundchannels[\maxrounds-1] \subset \cdots \subset \roundchannels[1] = [k]$, i.e., the algorithm proceeds until only one channel is left. Denote by \timeroundstart the time at which round $\roundidx$ starts. For all $t\in [\timeroundstart,\timeroundstart[{\roundidx+1}]-1]$ the actions allowed by $\ourpolicy$ are of the form $A_t = (X_t,\channelidx_t)$ where $\channelidx_t \in \roundchannels$. 

At every round $\roundidx\in[\maxrounds]$, the algorithm \ourpolicy senses each of the channels in $\roundchannels$ for a pre-specified number of times $\pullsround$ by sending each input symbol $x\in \cX$ an equal number of times (up to integer rounding), and observes the corresponding channel outputs to obtain an estimate $\capacitychannelestr$ of the capacity $\capacitychannel$ for all $\channelidx \in \roundchannels$, as explained in \cref{sec:confidence}.

After each round, assuming that the best channel is within $\roundchannels$, \ourpolicy uses an $(\epsci,\delta)$-PAC subroutine with parameters $(\epsround/2, \deltaround)$ computed as in~\cref{alg:bai} %
that senses the channels in $\roundchannels$ and outputs a channel $\pacchannelround \in \roundchannels$ that satisfies $\Pr\left(\suboptimalitygap[\pacchannelround] \leq \frac{\epsround}{2}\right) \geq 1-\deltaround$. We will later provide \extstart two subroutines suited for approximate channel identification and analyze their sample complexities. The difference between the subroutines is that one has a better asymptotic (in the number of channels) sample complexity. In contrast, the other has a better sample complexity for a small number of channels.\extend %

Having $\pacchannelround$, the set $\roundchannels [\roundidx+1]$ is constructed by comparing the estimates of the capacities for all $\channelidx \in \roundchannels$ to that of $\pacchannelround$ and discarding all the channels whose estimated capacity is less than $\capacitychannelestr[\pacchannelround] - \epsround$, i.e.,  $$\roundchannels[\roundidx+1]=\roundchannels \setminus \inbracelets{\channelidx \in \roundchannels \colon \capacitychannelestr < \capacitychannelestr[\pacchannelround] - \epsround}.$$

\ourpolicy is summarized in \cref{alg:bai}. The number of times each channel $\channelidx \in \roundchannels$ is sampled (sensed) in round $\roundidx$ is a function of $\confidence_r$ (and implicitly of $\card{\symtx}$, $\card{\symrx}$) given by
\extstart
\begin{align}
    \pullsround \geq 4 \linfac \nfactor[\deltaround]/\epsround^2 \log^2\big(4 \logfacvar[\deltaround] \nfactor[\deltaround]/\epsround^2\big), \label{eq:pullsround}
\end{align}
\extend
where $\epsround = 2^{-\roundidx}/4$ and $\deltaround = \confidence/(50\roundidx^3)$. The choice of $\pullsround$ is determined by the results for capacity estimation in \cref{sec:confidence}.
\setlength{\textfloatsep}{5pt}
\begin{algorithm}[t]
\caption{\ourpolicy Best Channel Identification}\label{alg:bai}
\begin{algorithmic}
\Require 0 < $\confidence$ < 1
\Initialize $\roundchannels[1] \gets [\nchannels]$, $\roundidx=1$, $\pullsround$ as in~\eqref{eq:pullsround}
\While{$\card{\roundchannels[\roundidx]} > 1$}
\State $\epsround = 2^{-\roundidx}/4$, $\deltaround = \confidence/(50\roundidx^3)$
\State Sample each channel $\channelidx \in \roundchannels$ for $\pullsround$ times and let $\capacitychannelestr$ be the resulting capacity estimate
\State $\pacchannelround \gets $ ($\epsround/2, \deltaround$)-PAC subroutine %
with channels $\roundchannels$ %
\State $\roundchannels[\roundidx+1] \gets \roundchannels \setminus \inbracelets{\channelidx \in \roundchannels: \capacitychannelestr < \capacitychannelestr[\pacchannelround] - \epsround}$
\State Set $\roundidx \gets \roundidx + 1$
\EndWhile
\Output $\{\outputchannel\} = \roundchannels$
\end{algorithmic}
\end{algorithm}
The correctness of $\ourpolicy$ and its stopping time reflected by its {sample complexity} are stated in \cref{thm:correctness_complexity}.

\begin{theorem} \label{thm:correctness_complexity}
For an appropriately chosen PAC subroutine, the policy \ourpolicy described in \cref{alg:bai} outputs the channel $\channelidx^\star$ with the largest capacity, i.e., $\channelidx^\star =  \argmax_{\channelidx \in [k]}\capacitychannel$, with probability at least $1-\delta$. 
The sample complexity of \ourpolicy required for determining the best channel with confidence $1-\delta$ is given by
    \begin{align*}
         \mathcal{O} \Bigg( \sum_{j \in [\nchannels]} \suboptimalitygap^{-2} &\log\left(\frac{1}{\delta}\log\left(\suboptimalitygap^{-1}\right)\right) \\[-.3cm]
         &\underbrace{\left( \log^2\log\left(\frac{1}{\delta}\log\left(\suboptimalitygap^{-1}\right)\right) +  \log^2\left(\suboptimalitygap^{-1}\right) \right)}_{\text{additional term compared to rewards in $[0,1]$ as in \cite{karnin2013almost}}} \Bigg).
    \end{align*}
\end{theorem}
The dependence of the sample complexity on the alphabet sizes $\card{\symtx}$ and $\card{\symrx}$ is hidden in the $\mathcal{O}$-notation (Bachmann–Landau notation). However, it is important to note that the dominating factor of the sample complexity scales with \extstart $\card{\symtx} \card{\symrx}$. \extend In comparison to the well-known gap-elimination algorithm of \cite{karnin2013almost} for arms with bounded rewards in $[0,1]$, the cost of estimating capacity in terms of channel senses is the additional multiplicative term highlighted in \cref{thm:correctness_complexity}.

\extstart
To analyze the actual non-asymptotic sample complexity, we plot in \cref{fig:delta_dependency} numerical simulations of \ourpolicy for a random set of binary DMCs and different levels of confidence. Out of $1024$ DMCs, we create four settings characterized by the combination of $10$ or $20$ DMCs and different minimal suboptimality gaps of $\Delta_{\min} \approx 0.08$ or $\Delta_{\min} \approx 0.13$, where $\Delta_{\min} \define \min_{\channelidx \in [\nchannels] \setminus \channelidx^\star} \suboptimalitygap$. We use two different values for $\Delta_{\min}$ to represent the hardness of the problem. It can be seen that increasing $\nchannels$ or decreasing $\Delta_{\min}$ results in an increasing number of channel senses, where doubling $\nchannels$ from $10$ to $20$ is less harmful than decreasing $\Delta_{\min}$ by less than a factor of $2$.
\extend
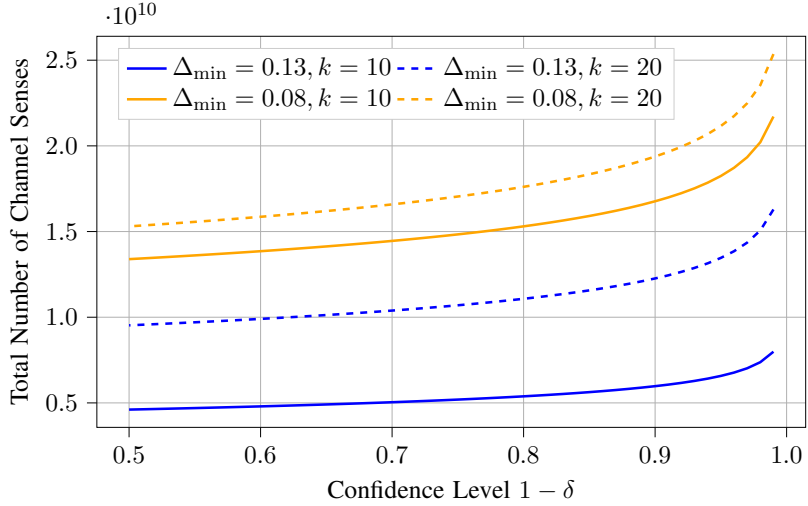
\begin{figure}[t]
    \centering
    \resizebox{.6\linewidth}{!}{\begin{tikzpicture}

\definecolor{darkgray176}{RGB}{176,176,176}
\definecolor{lightgray204}{RGB}{204,204,204}
\definecolor{orange}{RGB}{255,165,0}

\newcommand{\lw}{1.1}

\begin{axis}[
legend cell align={left},
legend columns=2,
legend style={
  fill opacity=0.8,
  draw opacity=1,
  text opacity=1,
  at={(0.03,0.97)},
  anchor=north west,
  draw=lightgray204
},
tick align=outside,
tick pos=left,
x grid style={darkgray176},
xlabel={Confidence Level \(\displaystyle 1-\delta\)},
xmajorgrids,
xmin=0.4755, xmax=1.0145,
xtick style={color=black},
xtick={0.4,0.5,0.6,0.7,0.8,0.9,1,1.1},
xticklabels={
  \(\displaystyle {0.4}\),
  \(\displaystyle {0.5}\),
  \(\displaystyle {0.6}\),
  \(\displaystyle {0.7}\),
  \(\displaystyle {0.8}\),
  \(\displaystyle {0.9}\),
  \(\displaystyle {1.0}\),
  \(\displaystyle {1.1}\)
},
y grid style={darkgray176},
ylabel={Total Number of Channel Senses},
ymajorgrids,
ymin=3572406307.3, ymax=26395335946.7,
ytick style={color=black},
ytick={0,5000000000,10000000000,15000000000,20000000000,25000000000,30000000000},
yticklabels={
  \(\displaystyle {0.0}\),
  \(\displaystyle {0.5}\),
  \(\displaystyle {1.0}\),
  \(\displaystyle {1.5}\),
  \(\displaystyle {2.0}\),
  \(\displaystyle {2.5}\),
  \(\displaystyle {3.0}\)
},
]
\addplot [thick, blue, line width=\lw]
table {%
0.99 7987161740
0.98 7375723770
0.97 7020369640
0.96 6769331220
0.95 6575255550
0.94 6417113980
0.93 6283715670
0.92 6168393780
0.91 6066855280
0.9 5976173290
0.89 5894262960
0.88 5819586770
0.87 5750978400
0.86 5687532170
0.85 5628530810
0.84 5573396600
0.83 5521657390
0.82 5472922390
0.81 5426864490
0.8 5383207140
0.79 5341714460
0.78 5302183650
0.77 5264439060
0.76 5228327580
0.75 5193714940
0.74 5160482730
0.73 5128526020
0.72 5097751360
0.71 5068075220
0.7 5039422560
0.69 5011725780
0.68 4984923750
0.67 4958960990
0.66 4933787040
0.65 4909355820
0.64 4885625170
0.63 4862556440
0.62 4840114040
0.61 4818265210
0.6 4796979650
0.59 4776229340
0.58 4755988300
0.57 4736232360
0.56 4716939060
0.55 4698087460
0.54 4679657970
0.53 4661632320
0.52 4643993350
0.51 4626725010
0.5 4609812200
};
\addlegendentry{$\Delta_{\min} = 0.13, k=10$}
\addplot [thick, blue, dashed, line width=\lw]
table {%
0.99 16288625600
0.98 15064781800
0.97 14353480760
0.96 13850970460
0.95 13462475160
0.94 13145905720
0.93 12878863220
0.92 12648003540
0.91 12444733800
0.9 12263195560
0.89 12099215700
0.88 11949716680
0.87 11812363900
0.86 11685344520
0.85 11567222780
0.84 11456842280
0.83 11353257900
0.82 11255687380
0.81 11163475940
0.8 11076070020
0.79 10992997420
0.78 10913852220
0.77 10838282780
0.76 10765982620
0.75 10696682980
0.74 10630146800
0.73 10566164020
0.72 10504547620
0.71 10445130360
0.7 10387762020
0.69 10332307280
0.68 10278643780
0.67 10226660440
0.66 10176256260
0.65 10127339000
0.64 10079824260
0.63 10033634660
0.62 9988698960
0.61 9944951560
0.6 9902331820
0.59 9860783660
0.58 9820255060
0.57 9780697620
0.56 9742066400
0.55 9704319440
0.54 9667417600
0.53 9631324220
0.52 9596005000
0.51 9561427780
0.5 9527562360
};
\addlegendentry{$\Delta_{\min} = 0.13, k=20$}
\addplot [thick, orange, line width=\lw]
table {%
0.99 21711945791
0.98 20210538204
0.97 19337096449
0.96 18719640336
0.95 18242037801
0.94 17852694529
0.93 17524144209
0.92 17240019357
0.91 16989777772
0.9 16766230444
0.89 16564255536
0.88 16380075158
0.87 16210823028
0.86 16054273070
0.85 15908662013
0.84 15772569476
0.83 15644834491
0.82 15524495989
0.81 15410749449
0.8 15302914617
0.79 15200411219
0.78 15102740295
0.77 15009469668
0.76 14920222572
0.75 14834668552
0.74 14752516191
0.73 14673507199
0.72 14597411560
0.71 14524023605
0.7 14453158611
0.69 14384650131
0.68 14318347626
0.67 14254114499
0.66 14191826454
0.65 14131370022
0.64 14072641366
0.63 14015545253
0.62 13959994040
0.61 13905906997
0.6 13853209499
0.59 13801832504
0.58 13751711983
0.57 13702788408
0.56 13655006418
0.55 13608314386
0.54 13562664075
0.53 13518010416
0.52 13474311148
0.51 13431526689
0.5 13389619823
};
\addlegendentry{$\Delta_{\min} = 0.08, k=10$}
\addplot [thick, orange, dashed, line width=\lw]
table {%
0.99 25357930054
0.98 23540993708
0.97 22484435812
0.96 21737745608
0.95 21160310558
0.94 20689671670
0.93 20292583644
0.92 19949237644
0.91 19646876230
0.9 19376800614
0.89 19132813788
0.88 18910345210
0.87 18705927408
0.86 18516867622
0.85 18341033034
0.84 18176705694
0.83 18022481448
0.82 17877198008
0.81 17739882444
0.8 17609712172
0.79 17485985490
0.78 17368099116
0.77 17255530464
0.76 17147824010
0.75 17044580262
0.74 16945446868
0.73 16850111568
0.72 16758296340
0.71 16669752550
0.7 16584256888
0.69 16501608164
0.68 16421624438
0.67 16344140574
0.66 16269006368
0.65 16196084674
0.64 16125249994
0.63 16056387240
0.62 15989390486
0.61 15924162164
0.6 15860612100
0.59 15798656876
0.58 15738219118
0.57 15679226864
0.56 15621613246
0.55 15565315862
0.54 15510276478
0.53 15456440576
0.52 15403757132
0.51 15352178302
0.5 15301659150
};
\addlegendentry{$\Delta_{\min} = 0.08, k=20$}
\end{axis}

\end{tikzpicture}}
    \caption{\extstart We simulate \ourpolicy with two different numbers $\nchannels \in \{10, 20\}$ of randomly generated binary DMCs (solid and dashed lines) and minimum suboptimality caps $\Delta_{\min} = \min_{\channelidx \in [\nchannels] \setminus \channelidx^\star} \suboptimalitygap \in \{0.08, 0.13\}$ (blue and orange) for different values of $\delta$ between $0.5$ and $1$. \extend}
    \label{fig:delta_dependency}
\end{figure}

\begin{proof}[Proof of \cref{thm:correctness_complexity}]
We start by proving the sample complexity in \cref{thm:correctness_complexity}. The proof of correctness is given in \cref{app:proof_correctness}. %
To that end, we need \cref{def:partitions} and the intermediate results stated in \cref{lemma:tr,lemma:trs,lemma:trinf}. We prove \cref{thm:correctness_complexity} afterward.
\newcommand{\channelpartn}{\ensuremath{a_s}}
\begin{definition} \label{def:partitions}
Let $\suboptimalitygap[\min] \define \min_{\channelidx: \suboptimalitygap[\channelidx] \neq 0} \suboptimalitygap[\channelidx]$ be the minimum suboptimality gap and $\suboptimalitygap[\max] \define \max_{\channelidx} \suboptimalitygap[\channelidx]$ the maximum. Then, we partition the channels for all integers $\lfloor \log_2(1/\suboptimalitygap[\max]) \rfloor \leq s \leq \lceil\log_2(1/\suboptimalitygap[\min])\rceil$ into sets $\channelpart$ defined as
\begin{align*}
    \channelpart \define \inbracelets{\channelidx \in [\nchannels]: 2^{-s} \leq \suboptimalitygap[\channelidx] < 2^{-s+1}}.
\end{align*}
We further define $\channelpartn \define \card{\channelpart}$ and $\channelpartround \define \roundchannels \cap \channelpart$.
\end{definition}

\begin{lemma} \label{lemma:tr}
Choose $\pullsround$ to satisfy \eqref{eq:pullsround} with equality. We have
\begin{align*}
    \sum_{r=1}^{s-1} \pullsround = \mathcal{O} \left( s^2 4^s \log\left(\frac{s}{\delta}\right) + 4^s \log\left(\frac{s}{\delta}\right) \log^2\log\left(\frac{s}{\delta}\right) \right).
\end{align*}
\end{lemma}
\begin{proof}
    The proof is given in \cref{app:proof_lemma_tr}.
\end{proof}
\begin{lemma} \label{lemma:trs}
Given an integer $s\geq1$ and $\pullsround[\roundidx+s]$ chosen to satisfy \eqref{eq:pullsround} with equality, we have
\begin{align*}
    \sum_{r=0}^\infty \inpara{\frac{1}{8}}^{r+1} \pullsround[\roundidx+s] \! = \! \mathcal{O} \left( s^2 4^s \log\left(\frac{s}{\delta}\right) + 4^s \log\left(\frac{s}{\delta}\right)  \log^2\log\left(\frac{s}{\delta}\right)\right).
\end{align*}
\begin{proof}
    The proof is given in \cref{app:proof_lemma_trs}.
\end{proof}
\end{lemma}

\begin{lemma} \label{lemma:trinf}
Let $\pullsround[\roundidx]$ be chosen to satisfy \eqref{eq:pullsround} with equality, then
\begin{align*}
    \sum_{r=1}^\infty \inpara{\frac{1}{8}}^{r-1} \!\!\!\!\!\!\! \pullsround[\roundidx] =  \mathcal{O} \left( \log\left(\frac{1}{\delta}\right)  \log^2\log\left(\frac{1}{\delta}\right)\right).
\end{align*}
\begin{proof}
    The proof is given in \cref{app:proof_lemma_trinf}.
\end{proof}
\end{lemma}

We are now ready to prove \cref{thm:correctness_complexity}, for which we follow \cite[Lemma 3.5]{karnin2013almost} to bound the number of channel sensing operations. \extstart To prove the asymptotic sample complexity, we use as a PAC-subroutine in \cref{alg:bai} the \medianpac policy introduced and analyzed in the following section. The number of operations required by the subroutine (cf. \cref{thm:correctness-complexity-pac-median}) are in the order of $\pullsround$ (cf. \eqref{eq:pullsround}), and henceforth can be ignored for the analysis. For $s\geq1$, the number of times $T_s$ the learner senses channels from $\channelpart$ \extstart can be bounded by using the fact that from round $s$ onward, the cardinality of the candidate set $\channelpartround$ will decrease (w.h.p.) by a fraction of $\frac{1}{8}$ per round, hence the number of per-round senses, accordingly. Up until round $s$, we assume the worst-case scenario of a non-decreasing candidate set described by $\channelpart$. In particular, we can write \extend \vspace{-0.2cm}
\begin{align*}
    T_s &= \sum_{r=1}^\infty \vert \channelpartround \vert \pullsround \leq \sum_{r=1}^{s-1} \vert \channelpart\vert \pullsround + \sum_{r=s}^\infty \vert \channelpartround \vert \pullsround \nonumber \\[-0.15cm]
    &\leq \channelpartn \sum_{r=1}^{s-1} \pullsround + \channelpartn \sum_{r=0}^\infty \inpara{\frac{1}{8}}^{r+1} \!\!\! \pullsround[\roundidx+s], \\[-0.6cm] \nonumber
\end{align*}
\extstart which we can bound  using \cref{lemma:tr,lemma:trs}. 
For $s < 1$, \extstart meaning for large suboptimality gaps, the cardinality of the candidate set will decrease (w.h.p.) from the first round onward. Hence, \extend we have $T_s = \channelpartn \sum_{r=1}^\infty \inpara{\frac{1}{8}}^{r-1} \pullsround[\roundidx]$, \extstart which we can bound by the help of \cref{lemma:trinf}. \extend
\extstart We can further observe that \extend for all channels $\channelidx \in \channelpart$, by \cref{def:partitions} we have $2^s < \frac{2}{\suboptimalitygap}$. Hence, from  \cref{lemma:tr,lemma:trs,lemma:trinf}, we obtain $T_s$ for all $s$ as
\begin{align*}
    T_s 
    &= \mathcal{O} \left( \channelpartn 4^s \log\left(\frac{s}{\delta}\right) \left( \log^2\log\left(\frac{s}{\delta}\right) + s^2 \right)\right) \\[-0.1cm]
    &= \begin{aligned}[t] \mathcal{O} \Bigg( \sum_{\channelidx\in \channelpart} &\suboptimalitygap^{-2} \log\left(\frac{1}{\delta}\log\left(\suboptimalitygap^{-1}\right)\right) \\[-0.3cm]
         &\left( \log^2\log\left(\frac{1}{\delta}\log\left(\suboptimalitygap^{-1}\right)\right) +  \log^2\left(\suboptimalitygap^{-1}\right) \right) \Bigg). \end{aligned} \\[-0.7cm]
\end{align*}
Summing over all $\channelpart$ (all channels) concludes the proof.
\end{proof}

\begin{remark}
    If a capacity-achieving input distribution $P^\star_{\rvX}\in{\cP}({\cX})$ is known for each channel $\channelidx$ (or only for the best channel $\channelidx^\star$), or even fixed due to some constraints, then the problem of finding the channel with maximum capacity boils down to determining the channel with the largest mutual information for $P^\star_{\rvX}$. In this case, the required number of channel senses is obtained by replacing $\nfactor[\alpha]$ in~\eqref{eq:req_samples} by $\nfactorprime[\alpha] \define 2 (3\card{\symtx} +2)^2 \log\inpara{4/\alpha}$. This follows from using the confidence bound for mutual information estimation established in \cite{shamir2010learning} and following the same derivations as for obtaining \cref{thm:correctness_complexity}.
\end{remark}

\extstart
\section{Approximate Best Channel Identification} \label{sec:pac}
By approximate best channel identification, we consider the problem of identifying with high probability an \textit{$\varepsilon$-best channel} $\channelidx_\varepsilon \in \{\channelidx: \suboptimalitygap\leq \varepsilon\}_{\channelidx \in [\nchannels]}$ such that the identified channel $\pacoutputchannel$ satisfies $$\Pr\left(\suboptimalitygap[\pacoutputchannel] \leq \epsci\right) \geq 1-\delta.$$ An algorithm that fulfills this property is called $(\epsci, \delta)$-PAC. A variety of such algorithms have been proposed and analyzed, most notably in \cite{evendar2006action}. However, the required assumptions on the bandits' reward distributions are not satisfied in the specific case of best channel identification. In this section, we will focus on two improved algorithms studied in  works subsequent to  \cite{evendar2006action}, and adapt them to the channel identification problem. After this modification, each of these algorithms can be used either as a standalone $(\epsci, \confidence)$-PAC policy, or as a subroutine in \cref{alg:bai}.

To build an intuition about how this class of algorithms work, we will first introduce a naive sampling strategy for channel identification, which has been studied for rewards with bounded support in \cite{evendar2006action}. Showing the correctness and the sample complexity of this algorithm is simple. The algorithm, which we term \pacpolicy, \extend
senses each channel $\channelidx \in \mathcal{C}$ for $\nsamples$ times to estimate $\capacitychannelestn$ as explained \extstart in \cref{sec:confidence}\extend. The policy then outputs the channel $\pacoutputchannel$ that \extstart maximizes $\capacitychannelestn$ and \extend is \extstart shown \extend to be $\epsci$-close to the best channel $\channelidx^\star$ with probability at least $1-\delta$, i.e., $\Pr\big(\suboptimalitygap[\pacoutputchannel] \leq \epsci\big) \geq 1-\delta$. \pacpolicy is summarized in \cref{alg:pac-bai}. %
The number of required channel senses (per channel) follows from \cref{thm:req_samples} and is given by
\extstart
\begin{align}\label{eq:naivepulls}
    \nsamples \geq \max\inbracelets{\frac{4 \linfac \nfactor[\frac{\delta}{2\nchannels}]}{\epsci^2} \log^2\left(\frac{4 \logfacvar[\frac{\delta}{2\nchannels}] \nfactor[\frac{\delta}{2\nchannels}]}{\epsci^2}\right), \frac{2 \nfactor[\frac{\delta}{2\nchannels}]}{\epsci}}.
\end{align}
\extend
\vspace{-0.3cm}

\begin{algorithm}[t]
\caption{$(\epsci, \delta)$-PAC \pacpolicy}%
\label{alg:pac-bai}
\begin{algorithmic}
\Require 0 < $\delta$ < 1, $\epsci > 0$, Set of channels $\mathcal{C}$,  $\nsamples$ as in~\eqref{eq:naivepulls}
\State Sample each channel $\channelidx \in \roundchannels$ for $\nsamples$ times and let $\capacitychannelestn$ be the resulting capacity estimates
\Output $\{\pacoutputchannel\} = \argmax_{\channelidx \in \mathcal{C}} \capacitychannelestn$
\end{algorithmic}
\end{algorithm}

\extstart We continue to analyze the properties of \cref{alg:pac-bai}, i.e., the correctness and the sample complexity, which are summarized in \cref{thm:correctness-complexity-pac}.\extend

\begin{theorem} \label{thm:correctness-complexity-pac}
The $(\epsci, \delta)$-PAC \pacpolicy described in \cref{alg:pac-bai} outputs a channel $\pacoutputchannel\in \mathcal{C}$ with $\suboptimalitygap[\pacoutputchannel] \leq \epsci \leq 2$ with probability at least $1-\delta$.
The total required sample complexity is given by
\begin{align*}
     \mathcal{O} \Bigg( \frac{\lvert \mathcal{C}\rvert}{\epsci^2} \log\left(\frac{\lvert \mathcal{C}\rvert}{\delta}\right) \log^2\left(\frac{1}{\epsci^2}\log\left(\frac{\lvert \mathcal{C}\rvert}{\delta}\right)\right).
\end{align*}
\end{theorem}

We assume $\epsci \leq 2$ since we will use \pacpolicy as a subroutine in \ourpolicy, ensuring that $\epsci \leq 1$. 
\begin{proof}

Consider the \pacpolicy described in \cref{alg:pac-bai}. To prove the correctness of the algorithm, let $\channel \in \mathcal{C}$ be a channel with $\suboptimalitygap[\channelidx] > \epsci$. The algorithm chooses $\nsamples$ to satisfy \eqref{eq:naivepulls}. We have from \cref{thm:req_samples} that
\begin{align*}
\Pr\big(\capacitychannelestn > \capacitychannel + \epsci/2\big) &\leq \delta/(2\lvert \mathcal{C}\rvert), \\
\Pr\big(\capacitychannelestn[\channelidx^\star] < \capacitychannel[\channelidx^\star] - \epsci/2\big) &\leq \confidence/(2\lvert \mathcal{C}\rvert).
\end{align*}
Since $\suboptimalitygap[\channelidx] = \capacitychannel[\channelidx^\star] - \capacitychannel[\channelidx] > \epsci$, it thus follows that $$\Pr\big(\capacitychannelestn[\channelidx^\star] < \capacitychannelestn[\channelidx]\big) \leq \delta/\card{\mathcal{C}}.$$ By the union bound, we have $\Pr\big(\exists \channelidx: \suboptimalitygap > \epsci, \capacitychannelestn[\channelidx^\star] < \capacitychannelestn[\channelidx] \big) %
    \leq \delta$, which proves the correctness.

\extstart
Let now $\epsci \leq 2$, then for $\tmpx \define 4 \card{\symtx} \card{\symrx}$ a straightforward calculation bounds the total number $T$ of channel senses by \vspace{-0.1cm}
\begin{align*}
    T \leq \frac{4 \linfac \tmpx \card{\mathcal{C}} \log\big(\frac{2\card{\mathcal{C}} \card{\symtx}}{\delta}\big)}{\epsci^2} \log^2\bigg(\frac{4 \logfacvar \tmpx \card{\mathcal{C}} \log\big(\frac{2\card{\mathcal{C}} \card{\symtx}}{\delta}\big)}{\epsci^2}\bigg). \\[-0.6cm]
\end{align*}
\extend
This leads to the complexity stated in \cref{thm:correctness-complexity-pac}.
\end{proof}

\extstart Such naive strategies sample the candidates several times and select the most promising candidate. They are known to achieve good empirical results whenever the candidate set (i.e., the number of arms or channels) is small \cite{hassidim2020optimal}. However, for a larger set of candidates, the \texttt{MedianElimination} algorithm proposed in \cite{evendar2006action} is commonly used. This algorithm was shown to be an asymptotically optimal PAC-algorithm. \extend
For $(\epsci, \confidence)$-PAC problems with a reasonably small set of candidate arms (channels in our case) whose rewards are bounded in the interval $[0,1]$, \cite{hassidim2020optimal} shows that a naive sampling strategy as in \cref{alg:pac-bai} is indeed to be preferred over \texttt{MedianElimination}. Furthermore, the authors of \cite{hassidim2020optimal} also propose improved sampling strategies that are beneficial over naive sampling only when the number of arms is larger than $10^5$.

\begin{algorithm}[t]
\extstart
\caption{$(\epsci, \delta)$-PAC \medianpac}%
\label{alg:pac-median}
\begin{algorithmic}
\Require 0 < $\delta$ < 1, $\epsci > 0$, Set of channels $\mathcal{C}$
\State $\roundchannels[1] \gets [\nchannels]$, $\epsround[1] = \epsci/4$, $\deltaround[1] = \delta/2$, $\roundidx=1$, $\pullsround$ as in \eqref{eq:median_pac_senses}
\While{$\card{\roundchannels[\roundidx]} > 1$}
\State Sample each channel $\channelidx \in \roundchannels$ for $\pullsround$ times
\State Let $\capacitychannelestn$ be the resulting capacity estimates
\State Let $c_\roundidx$ be the channel with the median estimate
\State $\roundchannels[\roundidx+1] \gets \roundchannels \setminus \inbracelets{\channelidx \in \roundchannels: \capacitychannelestr < \capacitychannelestr[c_\roundidx]}$
\State $\epsci_{\roundidx+1} = 3\epsround/4$, $\delta_{\roundidx+1} = \deltaround/2$, $\roundidx \gets \roundidx + 1$
\EndWhile
\Output $\{\pacoutputchannel\} = \argmax_{\channelidx \in \mathcal{C}} \capacitychannelestn$, an $\epsci$-best channel with probability $1-\delta$.
\end{algorithmic}
\end{algorithm}

\extstart Indeed, we will show that \extend an adaptation of the well-known \texttt{MedianElimination} \cite{evendar2006action} has a better asymptotic complexity than \pacpolicy in terms of $\cO$-notation, but incurs large multiplicative overheads. \extstart We will introduce and analyze a modified version suited for our problem, which will also be used in the asymptotic sample complexity analysis of \ourpolicy. \extend
However, since the number of channels here is likely to be rather small, following the same line of argumentation of \cite{hassidim2020optimal}, we expect that for small $\nchannels$, \pacpolicy leads to smaller actual sample complexity, despite a slightly worse asymptotic complexity statement. \extstart We will verify this statement both theoretically and experimentally in the sequel.

We call the modified version of \texttt{MedianElimination} \medianpac. In the terminology of our paper, this round-based algorithm discards half of the candidate channels each round which are likely not to be $\epsci$-optimal. We refer to that algorithm as \medianpac and provide a summary in \cref{alg:pac-median}. To fulfill the requirements of an $(\epsci, \confidence)$-PAC policy, the number of channel senses per round is required to be
\begin{equation} \label{eq:median_pac_senses}
\pullsround \geq \max\inbracelets{\frac{4 \linfac \nfactor[\deltaround/3]}{\epsround^2} \cdot  \log^2\left(\frac{4 \logfacvar \nfactor[\deltaround/3]}{\epsround^2}\right), \frac{2 \nfactor[\deltaround/3]}{\epsround^2}}.
\end{equation}

\cref{thm:correctness-complexity-pac-median} further summarizes the most important properties of \medianpac.

\begin{theorem} \label{thm:correctness-complexity-pac-median}
The $(\epsci, \delta)$-PAC \medianpac described in \cref{alg:pac-median} outputs a channel $\pacoutputchannel\in \mathcal{C}$ with $\suboptimalitygap[\pacoutputchannel] \leq \epsci \leq 4$ with probability at least $1-\delta$.
The total required sample complexity is given by
\begin{align*}
     \mathcal{O} \left( \frac{\lvert \mathcal{C}\rvert}{\epsci^2} \log\left(\frac{1}{\delta}\right) \log^2\left(\frac{1}{\epsci^2}\log\left(\frac{1}{\delta}\right)\right) \right).
\end{align*}
\end{theorem}

\begin{proof}
    The proof is given in~\cref{sec:app_thm3}.
\end{proof}

One can observe that the asymptotic statement of sample complexity is better than that of \pacpolicy by the lack of the $\card{\mathcal{C}}$ factors in both $\log$ terms. To verify the hypothesis of \pacpolicy performing better for a small set of channels, we simulate the algorithms with a random set of binary DMCs for different $(\epsci, \delta)$ values and plot the resulting actual number of required channel senses in \cref{fig:pac_comparison}. We use $(\epsci=0.1, \delta=0.1)$ as a baseline and observe the impact of increasing either the former to $\epsci=0.3$ or the latter to $\delta=0.7$. As expected from \cref{thm:correctness-complexity-pac,thm:correctness-complexity-pac-median}, one can find that increasing $\epsci$ has a much larger impact on the number of channel senses than increasing $\delta$. Further, while increasing $\epsci$ does not impact the break-even point at which \medianpac gets beneficial over \pacpolicy, increasing $\delta$ shifts the break-even point to the left, i.e., towards a smaller number of channels.

\begin{figure}
    \centering
    \resizebox{.6\linewidth}{!}{\begin{tikzpicture}

\definecolor{darkgray176}{RGB}{176,176,176}
\definecolor{lightgray204}{RGB}{204,204,204}
\definecolor{orange}{RGB}{255,165,0}

\newcommand{\lw}{1.1}

\begin{axis}[
legend cell align={left},
legend columns=2,
legend style={
  fill opacity=0.8,
  draw opacity=1,
  text opacity=1,
  at={(0.03,0.97)},
  anchor=north west,
  draw=lightgray204
},
log basis x={10},
log basis y={10},
tick align=outside,
tick pos=left,
x grid style={darkgray176},
xlabel={Number of Channels},
xmajorgrids,
xmin=3.0314331330208, xmax=1351.17610063144,
xmode=log,
xtick style={color=black},
xtick={0.1,1,10,100,1000,10000,100000},
xticklabels={
  \(\displaystyle {10^{-1}}\),
  \(\displaystyle {10^{0}}\),
  \(\displaystyle {10^{1}}\),
  \(\displaystyle {10^{2}}\),
  \(\displaystyle {10^{3}}\),
  \(\displaystyle {10^{4}}\),
  \(\displaystyle {10^{5}}\)
},
y grid style={darkgray176},
ylabel={Total Number of Channel Senses},
ymajorgrids,
ymin=8835713.54118639, ymax=130073582883.292,
ymode=log,
ytick style={color=black},
ytick={100000,1000000,10000000,100000000,1000000000,10000000000,100000000000,1000000000000,10000000000000},
yticklabels={
  \(\displaystyle {10^{5}}\),
  \(\displaystyle {10^{6}}\),
  \(\displaystyle {10^{7}}\),
  \(\displaystyle {10^{8}}\),
  \(\displaystyle {10^{9}}\),
  \(\displaystyle {10^{10}}\),
  \(\displaystyle {10^{11}}\),
  \(\displaystyle {10^{12}}\),
  \(\displaystyle {10^{13}}\)
}
]
\addplot [thick, black]
table {%
4 178708122
8 416985618
16 893540610
32 1846650594
64 3752870562
128 7565310498
256 15190190370
512 30439950114
1024 60939469602
}; \label{style1}
\addplot [thick, black, dash pattern=on 1pt off 3pt on 3pt off 3pt, line width=\lw]
table {%
4 178708122
8 416985618
16 893540610
32 1846650594
64 3752870562
128 7565310498
256 15190190370
512 30439950114
1024 60939469602
}; \label{style2}
\addplot [thick, black, dashed, line width=\lw]
table {%
4 178708122
8 416985618
16 893540610
32 1846650594
64 3752870562
128 7565310498
256 15190190370
512 30439950114
1024 60939469602
}; \label{style3}
\addplot [thick, blue, mark=x, mark size=3, mark options={solid}, line width=\lw]
table {%
4 178708122
8 416985618
16 893540610
32 1846650594
64 3752870562
128 7565310498
256 15190190370
512 30439950114
1024 60939469602
};
\addplot [thick, orange, mark=x, mark size=3, mark options={solid}, line width=\lw]
table {%
4 149771912
8 343290032
16 774713712
32 1726839008
64 3810567424
128 8338679936
256 18119369984
512 39135559680
1024 84088565760
}; \label{line1}
\addplot [thick, blue, dash pattern=on 1pt off 3pt on 3pt off 3pt, mark=x, mark size=3, mark options={solid}, forget plot, line width=\lw]
table {%
4 96685320
8 225599080
16 483426600
32 999081640
64 2030391720
128 4093011880
256 8218252200
512 16468732840
1024 32969694120
}; \label{line2}
\addplot [thick, orange, dash pattern=on 1pt off 3pt on 3pt off 3pt, mark=x, mark size=3, mark options={solid}, forget plot, line width=\lw]
table {%
4 94150532
8 230977192
16 548239296
32 1270634624
64 2892341248
128 6491711488
256 14406243840
512 31674017792
1024 69100158976
};
\addplot [thick, blue, dashed, mark=x, mark size=3, mark options={solid}, forget plot, line width=\lw]
table {%
4 16257984
8 37935296
16 81289920
32 167999168
64 341417664
128 688254656
256 1381928640
512 2769276608
1024 5543972544
};
\addplot [thick, orange, dashed, mark=x, mark size=3, mark options={solid}, forget plot, line width=\lw]
table {%
4 13667648
8 31383712
16 70935520
32 158334368
64 349825024
128 766380928
256 1666992896
512 3603869184
1024 7750151168
};
\addplot[red,mark=o,mark size=3pt] coordinates {(64,3752870562)};
\label{beps}
\addplot[red,mark=o,mark size=3pt] coordinates {(6,159099080)}; %
\addplot[red,mark=o,mark size=3pt] coordinates {(50,267417664)}; %
\end{axis}

\node [draw,fill=white] at (rel axis cs: 0.59,-1.47) {\shortstack[l]{
\ref{style1} $(\epsilon=0.1, \delta=0.1)$ \\
\ref{style2} $(\epsilon=0.1, \delta=0.7)$ \\
\ref{style3} $(\epsilon=0.3, \delta=0.1)$}};

\node [draw,fill=white] at (rel axis cs: 0.05,-0.85) {\shortstack[l]{
\ref{line1} \medianpac \\
\ref{line2} \pacpolicy \\
\ref{beps} Break-Even}};

\end{tikzpicture}}
    \caption{\extstart Comparison of the proposed ($\epsci, \delta$)-PAC algorithms tested on different parameter sets ($\epsci, \delta$) with $10$ binary DMCs with random transition matrices $\channel, \channelidx \in [10]$. Solid lines mean ($\epsci=0.1, \delta=0.1$), dashed dotted lines ($\epsci=0.1, \delta=0.7$) and dashed lines ($\epsci=0.3, \delta=0.1$). \extend}
    \label{fig:pac_comparison}
\end{figure}
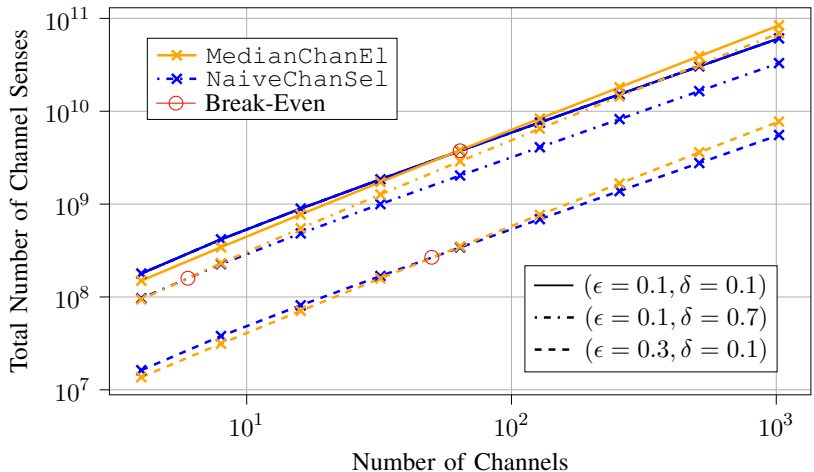

\section{An Information-Theoretic Lower Bound} \label{sec:converse}

This section establishes a lower bound on the expected sample complexity of any best channel identification algorithm with confidence $1-\delta$. To that end, we consider a bandit model %
described by the set of channels specified by conditional distributions $\dist \define (\channel[1], \dots, \channel[\change-1], \channel[\change], \channel[\change+1], \dots \channel[\nchannels])$, %
where for all $\channelidx \in [\nchannels]$ there exists a probability density $\channel \triangleq \channel(\rvY\vert \rvX)$ for which we require that $\channel \in \prmeasure$, where $\prmeasure$ is a suitable probability measure. To derive a lower bound, we consider a relaxed version of the best channel identification problem introduced above. In particular, we assume the learner knows the transition matrices of all channels $\channel$ beforehand; and hence, must correctly assign the channels to their transition matrices. Such a treatment of the problem facilitates enhanced sampling strategies that utilize the given knowledge of the environment, and provides a bound on the problem studied in this paper. To prove the lower bound, we rely on a change of measure argument that is based on altered versions of distribution $\dist$, where an element $\change \in [k]$ is replaced by an alternative channel. Let the altered distribution for any $\change \in [k]$ be $\dist^\prime_\change \define (\channel[1], \dots, \channel[\change-1], \channel[\change]^\prime, \channel[\change+1], \dots \channel[\nchannels])$. The channel $\channel[\change]^\prime$ is chosen such that it changes the desired probability of the event of identifying the best channel $\channelidx^\star$ under distribution $\dist$ of at least $1-\delta$ to at most $\delta$ under distribution $\dist_\change^\prime$. For this to happen, the best channel under $\dist_\change^\prime$ should not be $\channelidx^\star$. The following assumption formally describes the requirements for such altered distributions.

\begin{assumption} \label{assumption:alt_model}
    Consider a model $\dist$. Assume that for any $\change \in [\nchannels]\setminus \channelidx^\star$ there exists a model $\dist_\change^\prime$ with $W_{\change}^\prime \in \prmeasure$ such that for every input symbol $x \in \symtx$ it holds that $\kl{W_\change(\cdot\mid x)}{W_{\channelidx^\star}(\cdot\mid x)} < \kl{W_\change(\cdot\mid x)}{W_{\change}^\prime(\cdot\mid x)} < \kl{W_\change(\cdot\mid x)}{W_{\channelidx^\star}(\cdot\mid x)} + \kldiff$ and that $C(W_{\change}^\prime) > C(W_{\channelidx^\star})$ for any $\kldiff > 0$. %
    For $a=\channelidx^\star$, we further assume that there exists $W_{j^\star}^\prime \in \prmeasure$ such that for the suboptimal channel $\channelidx \in [\nchannels] \setminus \{\channelidx^\star\}$ and the corresponding input symbol $x_\channelidx \in \symtx$ that maximizes $\kl{W_{\channelidx^\star}(\cdot\mid x_\channelidx)}{W_{\channelidx}(\cdot\mid x_\channelidx)}$ it holds that $\kl{W_{\channelidx^\star}(\cdot\mid x_\channelidx)}{W_{\channelidx}(\cdot\mid x_\channelidx)} < \kl{W_{j^\star}(\cdot\mid x_\channelidx)}{W_{j^\star}^\prime(\cdot\mid x_\channelidx)} < \kl{W_{\channelidx^\star}(\cdot\mid x_\channelidx)}{W_{\channelidx}(\cdot\mid x_\channelidx)} + \kldiff$, where $C(W_{j^\star}^\prime) < C(W_\channelidx)$  for any $\kldiff > 0$. %
\end{assumption}

Informally, the above assumption ensures that for every suboptimal channel $\channelidx \in [\nchannels] \setminus \{\channelidx^\star\}$, there exists an alternative channel close to $\channelidx^\star$ with a larger capacity than $\channelidx^\star$. Similarly, for the best channel $\channelidx^\star$, there exists an alternative channel that is close to the furthest suboptimal channel $\channelidx$ (in terms of KL-divergence) with a smaller capacity than $\channelidx$.
With this at hand, we are now ready to state the best-case sample complexity.

\begin{theorem} \label{thm:lower_bound}
Considering a model $\dist$ satisfying \cref{assumption:alt_model}, the minimum number of required channel senses to identify the capacity-maximizing channel with probability at least $1-\delta$ is
\begin{align*}
    \mathbb{E}_\dist [\sigma] \geq &\sum_{a \in [\nchannels]\setminus \channelidx^\star} \frac{\log(1/2.4\delta)}{\max_{x \in \symtx} \kl{W_\change(\pi_y^\change \mid \pi_x^\change(x))}{W_{\channelidx^\star}(\cdot \mid x)}} \\
    &+ \frac{\log(1/2.4\delta)}{\max_{j \in [\nchannels] \setminus \channelidx^\star, x \in \symtx} \kl{W_{\channelidx^\star}(\cdot\mid x)}{W_{\channelidx}(\pi_y^\change \mid \pi_x^\change(x))}},
\end{align*}
where $\sigma$ is the stopping time of the algorithm, and $\pi_x^\change, \pi_y^\change$ are the permutations of input and output symbols for each suboptimal channel $\change \in [k] \setminus j^\star$ that maximize the lower bound on $\mathbb{E}_\dist [\sigma]$.
\end{theorem}

While the lower bound depends on the explicit statistical properties of the channels under consideration (i.e., their transition matrices), the sample complexities proved by our algorithms are functions of the suboptimality gaps. This discrepancy introduces a natural gap between the lower bound and the sample complexities stated above. Bounding the number of channel senses in terms of the suboptimality gaps $\suboptimalitygap$ from below would be a desired step but requires an upper bound on the quantity $\min_{\pi_y, x^\prime} \max_{x \in \symtx} \kl{W_\change(\pi_y \mid x^\prime)}{W_{\channelidx^\star}(\cdot \mid x)}$, and hence, on $\max_{x \in \symtx} \kl{W_{\channelidx^\prime}(\cdot\mid x)}{W_{\channelidx^\star}(\cdot\mid x)}$ for some permuted $W_\channelidx^\prime$ that minimizes the suboptimality gap. We detail in the following the lack of a direct relation between the maximal KL divergence and the suboptimality gap, and show that expressing our lower bound using suboptimality gap is impossible (at least, without additional assumptions). 

To that end, consider as the best among all DMCs in the candidate set $\mathcal{C}$ an approximation of a binary asymmetric channel (Z-channel) characterized by probability $q > 0$ and transition matrix $$\mathbf{W}_{\channelidx^\star} = \begin{pmatrix} 1-\varepsilon & \varepsilon \\ q & 1-q, \end{pmatrix},$$ where $\varepsilon \rightarrow 0$ is a vanishing constant. In addition, let channel with index $\channelidx \neq \channelidx^\star$ be a binary symmetric channel (BSC) with crossover probability $p > 0$ and transition matrix $$\mathbf{W}_{\channelidx} = \begin{pmatrix} p & 1-p \\ 1-p & p \end{pmatrix}.$$ Hence, when $\varepsilon \rightarrow 0$, the capacity of the best channel $\channelidx^\star$ approaches the capacity of a Z-channel, i.e., $\lim_{\varepsilon \rightarrow 0} \capacitychannel[\channelidx^\star] = \log\left(1 + (1-q)^{\frac{q}{1-q}}\right)  \in (0, \log(2)]$. Recall the capacity of the BSC $\channelidx$ being $\capacitychannel[\channelidx] = 1 - h_{\text{b}}(p) \in [0, \log(2))$. By the continuity of these two capacity formulas, for any $q$ and $\varepsilon \rightarrow 0$, there exists a $p$ parameterizing a suboptimal DMC $\channelidx$ with vanishing suboptimality gap, i.e., $\lim_{\varepsilon \rightarrow 0} \suboptimalitygap[\channelidx] = 0$ for some pair $(p, q)$. At the same time, the measure $\lim_{\varepsilon \rightarrow 0} \min_{\pi_y, x^\prime} \max_{x \in \symtx} \kl{W_\change(\pi_y \mid x^\prime)}{W_{\channelidx^\star}(\cdot \mid x)} = \infty$ goes to infinity.
Since naturally the upper bound on $\min_{\pi_y, x^\prime} \max_{x \in \symtx} \kl{W_\change(\pi_y \mid x^\prime)}{W_{\channelidx^\star}(\cdot \mid x)}$ should increase with $\suboptimalitygap$, the existence of such two channels contradicts the existence of a general bound that relates the divergence of individual rows in the transition matrices of the best and a suboptimal channel to the suboptimality gap of the latter. This finding raises the question of which figure of merit actually determines the difficulty of identifying channels with high capacity.

\begin{proof}[Proof of \cref{thm:lower_bound}]

We follow the proof method of \cite{kaufmann2016complexity}. First, recall that the action at time $t$ is $A_{t}=(X_{t},\channelidx_{t})$, i.e., the pair of a selected channel and an input symbol to be transmitted, and consider a certain $\change \in [k]$. We next define the following log-likelihood ratio of observing a certain set of actions and rewards under model $\dist$ and $\dist^\prime_\change$, where $\forall j \in [k]\setminus \{\change\}$, we set $W_j^\prime = W_j$:
\begin{align*}
    L_t &= L_t(A_1, \dots, A_t, Z_1, \dots, Z_t) \\
    &\define \sum_{\channelidx = 1}^\nchannels \sum_{\sampleidx=1}^t \sum_{x \in \symtx} \mathds{1}\{A_\sampleidx = (x, \channelidx)\} \log\left(\frac{\channel(Z_\sampleidx \vert x)}{\channel^\prime(Z_\sampleidx \vert x)}\right) \\
    &= \sum_{\channelidx = 1}^\nchannels \sum_{x \in \symtx} \sum_{\sampleidx=1}^{\nsenses} \log\left(\frac{\channel(\obssenses \vert x)}{\channel^\prime(\obssenses \vert x)}\right).
\end{align*}
The last equality holds from the i.i.d. assumption of the channel model considering $N_{x,j}$ as the number of times channel $\channelidx$ has been sensed with input symbol $x$ and $(\obssenses)$ as the sequence of subsequently observed symbols from that channel $\channelidx$. Next, applying Wald's Lemma \cite{wald1944cumulative} to the log-likelihood ratio $L_\sigma$ at stopping time $\sigma$ yields
\begin{align*}
    \mathbb{E}_\dist[L_t] &= \sum_{\channelidx = 1}^\nchannels \sum_{x \in \symtx} \mathbb{E}_\dist\left[\nsenses \right] \mathbb{E}_\dist\left[\log\left(\frac{\channel(\obssenses \vert x)}{\channel^\prime(\obssenses \vert x)}\right)\right] \\
    &= \sum_{\channelidx = 1}^\nchannels \sum_{x \in \symtx} \mathbb{E}_\dist\left[\nsenses \right] \kl{W_\channelidx(\cdot\mid x)}{W_\channelidx^\prime(\cdot\mid x)} \\
    &= \sum_{x \in \symtx} \mathbb{E}_\dist\left[\nsenses[\change] \right] \kl{W_\change(\cdot\mid x)}{W_\change^\prime(\cdot\mid x)},
\end{align*}
where the last equality holds since $W_j=W_j^\prime$ for all $j\neq\change$. 
To relate the expected value of the log-likelihood ratio considering the two different models $\dist$ and $\dist^\prime_\change$ to the corresponding probabilities of a certain event $\mathcal{E}$ upon stopping time $\sigma$, we make use of the following lemma from \cite{kaufmann2016complexity}:

\begin{lemma}[{\!\!\cite[Lemma 19]{kaufmann2016complexity}}]\label{lemma:llrevents}
    For every event $\mathcal{E} \in {\cF}_{\sigma}$, we have
    \begin{align*}
        \mathbb{E}_\dist[L_t] \geq d_b(\Pr_\dist(\mathcal{E}), \Pr_{\dist^\prime}(\mathcal{E})),
    \end{align*}
    where $d_b(x,y)$ is the binary relative entropy, and $\Pr_\dist(\mathcal{E})$ and $\Pr_{\dist^\prime}(\mathcal{E})$ are the probabilities of the event $\mathcal{E}$ under the models $\dist$ and $\dist^\prime$, respectively.
\end{lemma}

We use the result of \cref{lemma:llrevents} with the event of selecting the best channel $\channelidx^\star$ under $\dist$ as the best channel. By the requirements for every algorithm, the probability of this happening under model $\dist$ is bounded from below by $1-\delta$ and for $\dist^\prime_\change$ bounded from above by $\delta$. This is because, under the latter model, $j^\star$ is no longer the best channel. We start with bounding the number of channel senses of all suboptimal channels $\channelidx \in [k] \setminus \{j^\star\}$. Bounding the KL-divergence terms by the maximum over all $x \in \symtx$ to later obtain a bound on $\sum_{x \in \symtx} \mathbb{E}_\dist\left[\nsenses[\change] \right]$, we simplify to
\begin{align*}
        \sum_{x \in \symtx} &\mathbb{E}_\dist\left[\nsenses[\change] \right] \max_{x \in \symtx} \kl{W_\change(\cdot\mid x)}{W_\change^\prime(\cdot\mid x)} \\
        &\geq \sum_{x \in \symtx} \mathbb{E}_\dist\left[\nsenses[\change] \right] \kl{W_\change(\cdot\mid x)}{W_\change^\prime(\cdot\mid x)} \\
        &\geq \sup_{\mathcal{E}\in \cF_\sigma} d_b(\Pr_\dist(\mathcal{E}), \Pr_{\dist^\prime}(\mathcal{E})) \geq d_b(1-\delta, \delta) \geq \log(1/2.4\delta).
\end{align*}

Now using \cref{assumption:alt_model} and letting $\kldiff$ go to zero, we get for every suboptimal channel $\change \in [\nchannels]\setminus \channelidx^\star$ that
\begin{align*}
        \sum_{x \in \symtx} \mathbb{E}_\dist\left[\nsenses[\change] \right] &\geq \frac{\log(1/2.4\delta)}{\max_{x \in \symtx} \kl{W_\change(\cdot\mid x)}{W_\change^\prime(\cdot\mid x)}} \\
        &\geq \frac{\log(1/2.4\delta)}{\max_{x \in \symtx} \kl{W_\change(\cdot\mid x)}{W_{\channelidx^\star}(\cdot\mid x)} + \kldiff} \\
        &\overset{\kldiff \rightarrow 0}{=} \frac{\log(1/2.4\delta)}{\max_{x \in \symtx} \kl{W_\change(\cdot\mid x)}{W_{\channelidx^\star}(\cdot\mid x)}}.
\end{align*}

To tighten the lower bound, we can decouple the action $A_{t}=(X_{t},\channelidx_{t})$ from the assignment of an input symbol $x^\prime_t$ for the selected channel $\channelidx_{t}$, i.e., for each channel there exists a bijective function that maps the action $X_t$ to an input symbol $x^\prime_t$. We choose that function to maximize the lower bound, which is equivalent to considering all row permutations $\pi_x^\change$ of $W_\change$. Similarly, we can choose a mapping from observation $Y_t$ to an output symbol $Y_t^\prime$ to maximize the lower bound, which can be done by optimizing over all possible column permutations $\pi_y^\change$ of the channel matrix $W_\change$. Considering that we are optimizing over the choice of the input symbol, the optimization over a permutation of input symbols translates to optimizing over a specific choice of the symbol. For any choice of $\pi_x^\change$ and $\pi_y^\change$, which we detail in the sequel, we get
\begin{align}
        \sum_{x \in \symtx} \mathbb{E}_\dist\left[\nsenses[\change] \right] &\geq \frac{\log(1/2.4\delta)}{\max_{x \in \symtx} \kl{W_\change(\pi_y^\change \mid \pi_x^\change(x))}{W_{\channelidx^\star}(\cdot \mid x)}}. \label{eq:lb_part1}
\end{align}

Similarly to above, for the uniquely optimal channel $a = \channelidx^\star$, we consider the most impactful change of measure with respect to the pair of channel and corresponding input symbol that maximizes the KL-divergence compared with $\channelidx^\star$. This is captured by \cref{assumption:alt_model}. In a similar way to which the previous bound was derived, we obtain
\begin{align*}
        \sum_{x \in \symtx} &\mathbb{E}_\dist\left[\nsenses[j^\star] \right] \geq \frac{\log(1/2.4\delta)}{\max_{x \in \symtx} \kl{W_{j^\star} (\cdot\mid x)}{W_{j^\star}^\prime(\cdot\mid x)}}\\
        &\geq \frac{\log(1/2.4\delta)}{\max_{j \in [\nchannels] \setminus \channelidx^\star} \max_{x \in \symtx} \kl{W_{\channelidx^\star}(\cdot\mid x)}{W_{\channelidx}(\cdot\mid x)} + \kldiff} \\
        &\overset{\kldiff \rightarrow 0}{=} \frac{\log(1/2.4\delta)}{\max_{j \in [\nchannels] \setminus \channelidx^\star, x \in \symtx} \kl{W_{\channelidx^\star}(\cdot\mid x)}{W_{\channelidx}(\cdot\mid x)}}.
\end{align*}

To tighten this lower bound, we apply a similar argument as above. Considering any good choice of row and column permutations $\pi_x^\change$ and $\pi_y^\change$ of $W_\change$ for all $\change \in [\nchannels]\setminus \channelidx^\star$, we have the following bound:
\begin{align}
        \sum_{x \in \symtx} &\mathbb{E}_\dist\left[\nsenses[\channelidx^\star] \right] \geq \frac{\log(1/2.4\delta)}{\max_{\change \in [\nchannels] \setminus \channelidx^\star, x \in \symtx} \kl{W_{\channelidx^\star}(\cdot\mid x)}{W_{\channelidx}(\pi_y^\change \mid \pi_x^\change(x))}} \label{eq:lb_part2}
\end{align}
We combine the result for all suboptimal channels in \eqref{eq:lb_part1} with the result for the optimal channel in \eqref{eq:lb_part2} and obtain
\begin{align}
    \mathbb{E}_\dist [\sigma] \geq &\sum_{a \in [\nchannels]\setminus \channelidx^\star} \sum_{x \in \symtx} \mathbb{E}_\dist\left[\nsenses[\change] \right] + \sum_{x \in \symtx} \mathbb{E}_\dist\left[\nsenses[\channelidx^\star] \right]. \label{eq:lb_final}
\end{align}
It remains to find a good choice of the permutations $\{\pi_x^\change, \pi_y^\change\}_{\change \in [k]\setminus j^\star}$ that maximizes the overall lower bound. This can be done by maximizing \eqref{eq:lb_final} with respect to all possible permutations. In particular, we have
\begin{align*}
    \{\pi_x^{\change, \star}, \pi_y^{\change, \star}\}_{\change \in [k]\setminus j^\star} = \argmax_{\{\pi_x^\change, \pi_y^\change\}_{\change \in [k]\setminus j^\star}} & \frac{1}{\max_{\change \in [\nchannels] \setminus \channelidx^\star, x \in \symtx} \kl{W_{\channelidx^\star}(\cdot\mid x)}{W_{\channelidx}(\pi_y^\change \mid \pi_x^\change(x))}} \\
    & + \sum_{\change \in [k]\setminus j^\star}\frac{1}{\max_{x \in \symtx} \kl{W_\change(\pi_y^\change \mid \pi_x^\change(x))}{W_{\channelidx^\star}(\cdot \mid x)}}
\end{align*}
Using the above choice of $\{\pi_x^{\change, \star}, \pi_y^{\change, \star}\}_{\change \in [k]\setminus j^\star}$ yields the statement in \cref{thm:lower_bound}.

\end{proof}

\extend

\section{Conclusion}\label{sec:conclusion}
We considered the problem of identifying the channel with the maximal capacity \extstart out of $\nchannels$ DMCs by using their input-output samples. To this end, we established a confidence interval bound for capacity estimation using the minimax characterization of capacity. \extend
We formulated the problem as best arm identification in the MAB context, \extstart proposed algorithms in the fixed confidence setting, and analyzed their sample complexity both theoretically and through numerical simulations. We further introduced an information-theoretic explicit lower bound on the sample complexity of every best channel identification problem based on the channel properties. \extend Our proposed capacity estimation algorithm uniformly distributes the samples across the channel's input letters $\symtx$. Nonetheless, a small subset of the input alphabet may only support the capacity-achieving input distribution. %
By identifying this subset throughout the algorithm, the learner may instruct the transmitter to refrain from sampling input letters not in this subset. Analyzing such algorithms is challenging because the capacity-achieving input distribution is not necessarily unique, and standard capacity computation algorithms such as Blahut-Arimoto \cite{blahut1972computation,arimoto1972algorithm} tend to produce fully supported input distributions. Even more, the support of the capacity-achieving input distribution might not be stable with respect to estimation errors of the channel. This direction is thus left for future research. \extstart Further analysis of the figures of merit that reflect the problem's difficulty and finding a lower bound based on the suboptimality gaps to bridge the gap between the achieved sample complexities and the lower bound is also left for future investigations.

\appendix

In \cref{app:proof_kl_bound} we prove the result of \cref{lemma:kl_bound}, which was stated in \cref{sec:confidence} and is used in the proof of \cref{lem: capacity bias}. In \cref{app:proof_concentration_capacity}, we prove the confidence bound for capacity estimation stated in \cref{prop: Concentration for capacity}. In \cref{app:proof_req_samples}, we prove the result for the number of samples that suffice to achieve a certain confidence level when estimating the capacity of a DMC, which was stated in \cref{thm:req_samples} and is required to prove \cref{thm:correctness_complexity,thm:correctness-complexity-pac,thm:correctness-complexity-pac-median}. In \cref{app:proof_correctness}, we prove the correctness of \ourpolicy. In \cref{app:proof_lemma_tr,app:proof_lemma_trs,app:proof_lemma_trinf}, we provide the proofs for \cref{lemma:tr,lemma:trs,lemma:trinf}, which we used in \cref{sec:best_channel_id} to prove the sample complexity given in \cref{thm:correctness_complexity}. Lastly, in \cref{sec:app_thm3}, we prove the sample complexity of \medianpac.

\subsection{Proof of \cref{lemma:kl_bound}} \label{app:proof_kl_bound}
\extstart
\begin{proof}
We want to prove that for every $P_Y(y)$ in the simplex ${\cal P}(\symrx)$, there exists a distribution $Q_{Y}^{*}$ in the simplex ${\cal P}_{\eta}(\symrx)$ such that $\kl{P_{Y}}{Q_{Y}^{*}}$ is bounded by $2\card{\symrx} \eta \leq 1$. Let $Q_{Y}^{*}$ be a modified distribution of $P_Y(y)$ constructed as
\begin{align*}
    Q_{Y}^{*} &= \frac{1}{\xi} \max \{ P_Y(y), \eta \xi \}, %
\end{align*}
with a normalization constant $\xi$ that satisfies
\begin{align*}
    \xi = &\sum_{y \in \symrx} \max\{P_Y(y), \xi \eta\} \\
    &\leq \sum_{y \in \symrx} \left(P_Y(y) + \xi \eta\right) = 1 + \xi \eta \card{\symrx}.
\end{align*}
Consequently, we have $\xi \leq \frac{1}{1-\eta \card{\symrx}}$. Using that result and the above definition, we can write
    \begin{align*}
        \min_{Q_{Y}\in{{\cal P}_{\eta}({\cal Y})}} &\kl{P_{Y}}{Q_{Y}} \leq \kl{P_{Y}}{Q_{Y}^{*}} \\
        & = \sum_{y \in \symrx} P_Y (y) \log \frac{P_Y(y)}{Q_{Y}^{*}} \\
        & = \sum_{y \in \symrx} P_Y (y) \log \frac{\xi P_Y(y)}{\max\{P_Y(y), \xi \eta\}} \\
        & = \log(\xi) + \sum_{y \in \symrx} P_Y (y) \log \frac{P_Y(y)}{\max\{P_Y(y), \xi \eta\}} \\
        &= \log(\xi) + \sum_{y \in \symrx: \xi \eta > P_Y(y)} P_Y (y) \log \frac{P_Y(y)}{\xi \eta}\\
        &\leq \log(\xi) \leq \log \left(\frac{1}{1-\eta \card{\symrx}}\right) \leq \left(\frac{\eta \card{\symrx}}{1-\eta \card{\symrx}}\right) \\
        &\leq 2 \eta \card{\symrx}.
    \end{align*}
This concludes the proof.
\end{proof}
\extend

\subsection{Proof of \cref{prop: Concentration for capacity}} \label{app:proof_concentration_capacity}
To prove \cref{prop: Concentration for capacity} we will need
the following lemma, which is a standard bound on the empirical total
variation distance, and is derived here for completeness and convenience.
\begin{lemma}
\label{lem: empirical total variation bound}Let $Y_{i}\sim P_{Y}$
be independent and identically distributed, and $\hat{P}_{Y}^{n}$ be the empirical distribution based
on $n$ observed samples. %
Then, 
\[
d_{\text{\emph{TV}}}\left(\hat{P}_{Y}^{n},P_{Y}\right)\le \sqrt{\frac{4\extstart |{\cal Y}|\extend \log\frac{1}{\alpha}}{n}}
\]
with probability $1-\alpha$. 

\begin{proof} %
The total variation distance $d_{\text{TV}}(\hat{P}_{Y}^{n},P_{Y})$ satisfies
a bounded difference inequality with constant $1/n$ as a function
of $(Y_{1},\ldots,Y_{n})$, and so by McDiarmid's inequality \cite[Theorem 3.11]{van2014probability}
\[
\Pr\left[d_{\text{TV}}(\hat{P}_{Y}^{n},P_{Y})-\EE\left[d_{\text{TV}}(\hat{P}_{Y}^{n},P_{Y})\right]\geq\epsilon\right]\leq e^{-2n\epsilon^{2}},
\]
or 
\begin{align*}
0 & \leq d_{\text{TV}}\left(\hat{P}_{Y}^{n},P_{Y}\right)%
 \leq\EE\left[d_{\text{TV}}(\hat{P}_{Y}^{n},P_{Y})\right]+\sqrt{\frac{\log\frac{1}{\alpha}}{2n}}
\end{align*}
with probability $1-\alpha$. Furthermore, 
\extstart
\begin{align*}
 & \EE\left[d_{\text{TV}}(\hat{P}_{Y}^{n},P_{Y})\right] \\
 &=\EE\left[\sum_{y\in{\cal Y}}\left|\hat{P}_{Y}^{n}(y)-P_{Y}(y)\right|\right]\\
 & \overset{(a)}{\leq} \sum_{y\in{\cal Y}} \sqrt{\EE\left[\left|\hat{P}_{Y}^{n}(y)-P_{Y}(y)\right|^{2}\right]}\\
 & \overset{(b)}{=} \sum_{y\in{\cal Y}} \sqrt{\frac{1}{\nsamples} P_{Y}(y) \left(1- P_{Y}(y)\right) }\\
 & \overset{(c)}{\leq} \sqrt{ \left( \sum_{y\in{\cal Y}} \frac{P_{Y}(y)}{\nsamples} \right) \left( \sum_{y\in{\cal Y}} \left(1- P_{Y}(y)\right) \right) } \\
 & = \sqrt{ \frac{\card{\symrx} - 1}{\nsamples} } \leq \sqrt{\frac{\card{\symrx}}{\nsamples}},
\end{align*}
where $(a)$ follows from Jensen's inequality; $(b)$ follows from the variance of an unbiased estimator for a Bernoulli random variable; and $(c)$ follows from Cauchy-Schwarz inequality. %
The result then follows since $\sqrt{a}+\sqrt{b}\le\sqrt{2(a+b)}$ and so 
\begin{align*}
\sqrt{\frac{\log\frac{1}{\alpha}}{2n}}+\sqrt{\frac{|{\cal Y}|}{\nsamples}} & \leq\sqrt{\frac{\log\frac{1}{\alpha}+2\card{\symrx}}{\nsamples}}\\
 & \leq\sqrt{\frac{4\card{\symrx}\log\frac{1}{\alpha}}{n}}.
\end{align*}
\extend
\end{proof}
\end{lemma}
\begin{proof}[Proof of \cref{prop: Concentration for capacity}]
The idea is to relate $\mathrm{C}(\hat{W})$ to the estimate $\mathrm{C}_{\eta}^{\nsamples}(\hat{W})$ of the pseudo-capacity for a properly chosen $\eta=1/\nsamples$. The value of $\eta$ controls the
bias-variance trade-off. The bias is large for large $\eta$, yet
the estimation variance is lower. 

To this end, our first goal is to obtain a high probability bound
on the estimation error of $\mathrm{C}_{\eta}(W)$ from 
\[
\hat{\mathrm{C}}_{\eta}^{\nsamples}(W)=\min_{Q_{Y}\in{\cal P}_{\eta}({\cal Y})}\max_{x\in{\cal X}}D(\hat{W}_{Y|X=x}\mid\mid Q_{Y}).
\]
\extstart
From Lemma \ref{lem: empirical total variation bound}, it holds for
any $x\in{\cal X}$ that $\nsamples/|{\cal X}|$ samples lead to 
\begin{equation}
d_{\text{TV}}\left(\hat{W}_{Y|X}(y|x),W_{Y|X}(y|x)\right)\leq\sqrt{\frac{4|{\cal X}||{\cal Y}|\log\frac{|{\cal X}|}{\alpha}}{\nsamples}}=:\kappa,\label{eq: upper bound on TV for conditional distribution}
\end{equation}
\extend
with probability larger than $1-\alpha/|{\cal X}|$. By the union
bound, the same holds simultaneously for all $x\in{\cal X}$ with
probability $1-\alpha$. We next assume that this event holds, and
evaluate the error $|\mathrm{C}_{\eta}(\hat{W})-\mathrm{C}_{\eta}(W)|$.
To begin, we note that
\begin{align*}
D(&W_{Y|X=x}\mid\mid Q_{Y}) \nonumber \\
&=-H(W_{Y|X=x}) -\sum_{y\in{\cal Y}}W_{Y|X=x}(y)\log Q_{Y}(y), %
\end{align*}
and so for any $x\in{\cal X}$ and $Q_{Y}\in{\cal Q}_{\eta} \define {\cal P}_{\eta}({\cal Y})$ the
triangle inequality implies
\begin{align}
 & \left|D(\hat{W}_{Y|X=x}\mid\mid Q_{Y})-D(W_{Y|X=x}\mid\mid Q_{Y})\right|\nonumber \\
 & \leq\left|H(\hat{W}_{Y|X=x})-H(W_{Y|X=x})\right|\nonumber \\
 & \hphantom{==}+\sum_{y\in{\cal Y}}\left|\hat{W}_{Y|X=x}(y)-W_{Y|X=x}(y)\right|\left|\log Q_{Y}(y)\right|.\label{eq: KL empirical vs population}
\end{align}
The first term in (\ref{eq: KL empirical vs population}) is an entropy
difference term, and can be bounded from the TV bound on the entropy
difference in \cite[Problem 3.10]{csiszar2011information}. In the following, we use for any $x,y$ the notation $x \vee y \define \max\{x, y\}$. Note the elementary inequality on the binary entropy
function
\begin{align}
h_{\text{b}}(t) & :=t\log\frac{1}{t}+(1-t)\log\frac{1}{(1-t)}\nonumber \\
 & \leq t\log\frac{1}{t}+t,\label{eq: elementary bound on binary entropy}
\end{align}
which holds for any $t\in[0,1]$. 
Using the bound in \cite{zhang2007estimating,audenaert2007sharp} (Equations (4) and (11), respectively)
implies that for any $\hat{W}_{Y|X},W_{Y|X}\in{\cal P}({\cal Y})$
\begin{align}
 & \left|H(\hat{W}_{Y|X})-H(W_{Y|X})\right|\nonumber \\
 & \begin{aligned}[t] \leq\frac{1}{2}&d_{\text{TV}}\left(\hat{W}_{Y|X},W_{Y|X}\right)\log(|{\cal Y}|-1) \\
 & +h_{\text{b}}\left(\frac{1}{2}d_{\text{TV}}\left(\hat{W}_{Y|X},W_{Y|X}\right)\right)\nonumber \end{aligned} \\
 & \overset{{\scriptstyle (a)}}{\leq}\frac{1}{2}d_{\text{TV}}\left(\hat{W}_{Y|X},W_{Y|X}\right)\log(|{\cal Y}|-1)\nonumber \\
 & \hphantom{==}+\frac{1}{2}d_{\text{TV}}\left(\hat{W}_{Y|X},W_{Y|X}\right)\log\frac{1}{\frac{1}{2}d_{\text{TV}}\left(\hat{W}_{Y|X},W_{Y|X}\right)}\nonumber \\
 & \hphantom{==}+\frac{1}{2}d_{\text{TV}}\left(\hat{W}_{Y|X},W_{Y|X}\right)\nonumber \\
 & \overset{{\scriptstyle (b)}}{\leq}\frac{\kappa}{2}\log(|{\cal Y}|-1)+\frac{\kappa}{2}\log\frac{2}{\kappa}+\frac{\kappa}{2}\nonumber \\
 & \extstart \overset{{\scriptstyle (c)}}{\leq}\frac{\kappa}{2}\log\left(\frac{2|{\cal Y}|}{\kappa}\right) + \frac{\kappa}{2}, \label{eq: empirical entropy error}\extend
\end{align}
\extstart
where $(a)$ follows from the bound in (\ref{eq: elementary bound on binary entropy}),
$(b)$ follows since $d_{\text{TV}}(\hat{W}_{Y|X},W_{Y|X})\leq\kappa$ defined in
(\ref{eq: upper bound on TV for conditional distribution}) and noticing that total variation is always bounded by 2.
\extend

\extstart
\begin{remark}
    A tighter bound could be found by considering the relation between local variation and total variation distance as established in \cite{sason2013entropy}. However, the benefits thereof are limited in the regime of interest for this paper. %
\end{remark}

\begin{remark}
    Alternatively to bound the entropy difference using a Fano-type inequality as above, similarly to \cite{shamir2010learning}, we could utilize McDiramid's inequality to bound the maximum change a single replacement in the measurements could carry to the entropy estimation. Taking into account the estimation bias, one can show that $\left|H(\hat{W}_{Y|X})-H({W}_{Y|X}) \right| \leq \frac{\kappa}{\sqrt{2}} \log(\nsamples/\card{\symtx}) + \frac{\kappa^2}{2}$ with probability at least $1-\alpha/\card{\symtx}$, which is a looser bound than the one in \eqref{eq: empirical entropy error}.
\end{remark}
\extend

Next, the second term in (\ref{eq: KL empirical vs population}) is
bounded for any $Q_{Y}\in{\cal Q}_{\eta}$ as 
\begin{align}
 & \sum_{y\in{\cal Y}}\left|\hat{W}_{Y|X=x}(y)-W_{Y|X=x}(y)\right|\left|\log Q_{Y}(y)\right|\nonumber \\
 & \leq d_{\text{TV}}\left(\hat{W}_{Y|X=x},W_{Y|X=x}\right)\log\frac{1}{\eta}\leq\kappa\log\frac{1}{\eta}.\label{eq: empirical expected log q error}
\end{align}
Hence, inserting (\ref{eq: empirical entropy error}) and (\ref{eq: empirical expected log q error})
to (\ref{eq: KL empirical vs population}) results 
\begin{align*}
 \Big|D(\hat{W}_{Y|X=x}\mid\mid Q_{Y})-&D(W_{Y|X=x}\mid\mid Q_{Y})\Big|\\
 & \leq\frac{\kappa}{2}\log\left(\frac{2 e|{\cal Y}|}{\eta^2\kappa}\right).
\end{align*}
Consequently, the pseudo-capacity estimation error is similarly bounded
as
\begin{align}
 & \left|\mathrm{C}_{\eta}(\hat{W})-\mathrm{C}_{\eta}(W)\right|\nonumber \\
 & \extstart = \left|\min_{Q_{Y}\in{\cal Q}_{\eta}}\max_{x\in{\cal X}} D({W}_{Y|X=x}\mid\mid Q_{Y}) - \min_{Q_{Y}\in{\cal Q}_{\eta}}\max_{x\in{\cal X}} D(\hat{W}_{Y|X=x}\mid\mid Q_{Y})\right| \extend \nonumber \\
 & \leq\extstart \max_{Q_{Y}\in{\cal Q}_{\eta}} \extend \max_{x\in{\cal X}}\left|D({W}_{Y|X=x}\mid\mid Q_{Y})-D(\hat{W}_{Y|X=x}\mid\mid Q_{Y})\right|\nonumber \\
 & \leq\frac{\kappa}{2}\log\left(\frac{2 e|{\cal Y}|}{\eta^2\kappa}\right).\label{eq: estimation error on pseudo-capacity}
\end{align}
\extstart Recalling that $\logfacinst \define \sqrt[5]{\frac{\card{\symrx} e^2}{\card{\symtx}\log\frac{|{\cal X}|}{\alpha}}}$, \extend by the triangle inequality 
\begin{align*}
 & \left|\mathrm{C}(\hat{W})-\mathrm{C}(W)\right|=\\
 & =\left|\mathrm{C}(\hat{W})-\mathrm{C}_{\eta}(\hat{W})+\mathrm{C}_{\eta}(\hat{W})-\mathrm{C}_{\eta}(W)+\mathrm{C}_{\eta}(W)-\mathrm{C}(W)\right|\\
 & \leq\left|\mathrm{C}(\hat{W})-\mathrm{C}_{\eta}(\hat{W})\right|+\left|\mathrm{C}_{\eta}(\hat{W})-\mathrm{C}_{\eta}(W)\right|+\left|\mathrm{C}_{\eta}(W)-\mathrm{C}(W)\right|\\
 & \extstart \overset{{\scriptstyle (a)}}{\leq} 2 \eta |{\cal Y}|+\frac{\kappa}{2}\log\left(\frac{2 e|{\cal Y}|}{\eta^2\kappa}\right) + 2 \eta |{\cal Y}|\\
 & \extstart \overset{{\scriptstyle (b)}}{\leq} \frac{5}{4}\kappa \log(\logfacinst \nsamples)+4 \frac{|{\cal Y}|}{\nsamples}\\
 & \extstart \overset{{\scriptstyle (c)}}{\leq} \frac{5}{4} \sqrt{\frac{4 |{\cal X}||{\cal Y}|\log\frac{|{\cal X}|}{\alpha}\log^{2}(\logfacinst \nsamples)}{\nsamples}} +\frac{4|{\cal X}||{\cal Y}|\log\frac{|{\cal X}|}{\alpha}}{\nsamples}, \extend
\end{align*}
where $(a)$ holds from (\ref{eq: estimation error on pseudo-capacity})
and utilization of Lemma \ref{lem: capacity bias} twice (for $V\equiv\hat{W}$
and then for $V=W$), $(b)$ follows by using the definition of $\kappa$
and setting $\eta=\frac{1}{\nsamples}$ resulting in 
\extstart \begin{align*}
&\frac{\kappa}{4}\log\left(\frac{4 e^2|{\cal Y}|^{2}}{\eta^{4}\kappa^{2}}\right) = \frac{\kappa}{4} \log\left(\frac{\nsamples \card{\symrx} e^2}{\eta^4|{\cal X}|\log\frac{|{\cal X}|}{\alpha}}\right) \\
&= \frac{\kappa}{4} \log\left(\frac{\nsamples^5 \card{\symrx} e^2}{|{\cal X}| \log\frac{|{\cal X}|}{\alpha}}\right) = \frac{5 \kappa}{4} \log\left(\nsamples \sqrt[5]{\frac{\card{\symrx} e^2}{|{\cal X}| \log\frac{|{\cal X}|}{\alpha}}}\right),
\end{align*}\extend and $(c)$ holds by replacing
the value of $\kappa$, \extstart and slightly relaxing the inequality using the fact that $\frac{4\card{\symrx}}{n} \leq \frac{\kappa^2}{2}$, which holds since $\card{\symtx} \geq 2$ and $\log(\card{\symtx}/\alpha) > 1$ by the choice of $\alpha$. \extend
\end{proof}

\subsection{Proof of \cref{thm:req_samples}} \label{app:proof_req_samples}
To prove \cref{thm:req_samples} and thereby bound the number of samples $\nsamples$ required to achieve a certain confidence for estimating the capacity associated with a channel, \extstart we need a refinement of Lemma 15 from \cite{weinberger2022multi}, which we state and prove in the following. \extend

\extstart
\begin{lemma}[{Refinement of \cite[Lemma 15]{weinberger2022multi}}] 
\label{lemma:sample_sufficiency}
    If $n \geq 15 \log^2 (\tmpconst /y)/y$, then it holds that $\frac{\log^2(\tmpconst n)}{n} \leq y$.
    \begin{proof}
        The proof is a refinement of \cite{weinberger2022multi} for the case where $r=2$. For $n \geq 1$, $\frac{\log^2(\tmpconst n)}{n} \leq y$ has a unique maximum of $\frac{4\tmpconst}{e^2}$ at $n = \frac{e^2}{\tmpconst}$. For $y \in [0, \frac{4\tmpconst}{e^2}]$, by using $n = 15 \log^2 (\tmpconst /y)/y$, we get
        \begin{align*}
            \frac{\log^2(\tmpconst n)}{n} &= y \cdot \frac{\log^2((\tmpconst /y) 15 \log^2 (\tmpconst /y))}{15 \log^2 (\tmpconst /y)} \\
            &\leq y \cdot \sup_{y \in [0, \frac{4\tmpconst}{e^2}]} \frac{\log^2((\tmpconst /y) 15 \log^2 (\tmpconst /y))}{15 \log^2 (\tmpconst /y)} \\
            &\!\!\!\!\!\! \overset{y^\prime \define y/\tmpconst}{\leq} y \cdot \sup_{y^\prime = \in [0, \frac{4}{e^2}]} \frac{\log^2((1 /y^\prime) 15 \log^2 (1/y^\prime))}{15 \log^2 (1 /y^\prime)} \\
            &\leq y,
        \end{align*}
        with the value of the supremum determined numerically. Note that this result can be generalized to arbitrary values of $r$ in $\frac{\log^r(\tmpconst n)}{n} \leq y$.
    \end{proof}
\end{lemma}
\extend

\begin{proof}[Proof of \cref{thm:req_samples}]
To make sure the confidence bound holds, we require that each of the two summands in \cref{prop: Concentration for capacity} is lower or equal to $\varepsilon/2$. Recall that $\nfactor \define 4 \card{\symtx} \card{\symrx} \log(\card{\symtx}/\alpha)$. We first consider the left term, i.e., we require that
\extstart
\begin{align}
    \frac{5\sqrt{\nfactor}}{4} \frac{\log(\logfacinst \nsamples)}{\sqrt{\nsamples}} & \leq \frac{\epsci}{2} \nonumber \\
    &\!\!\!\!\!\!\!\!\!\!\!\!\text{(Equivalently)} \nonumber \\
    \frac{4 \epsci^2}{25 \nfactor} &\geq \frac{\log^2(\logfacinst n)}{n}. \!\! \label{eq:nonlin_samplebound}
\end{align}
Applying \cref{lemma:sample_sufficiency} to \eqref{eq:nonlin_samplebound} with $y = \frac{4 \epsci^2}{25 \nfactor}$ and $\tmpconst = \logfacinst$ leads to the first argument in \cref{thm:req_samples}, i.e., we require that $n \geq 
15 \cdot \frac{25 \nfactor}{4 \epsci^2}  \log^2 (\logfacinst \frac{25 \nfactor}{4 \epsci^2})$. \extend Imposing the same requirement on the right term, i.e., $\frac{\nfactor}{\nsamples} \leq \frac{\epsci}{2}$, yields the second argument in \cref{thm:req_samples}. Since both requirements have to hold, we take the maximum of the two results, which proves \cref{thm:req_samples}.
\end{proof}

\subsection{Proof of \cref{thm:correctness_complexity}} \label{app:proof_correctness}

The correctness of \cref{alg:bai} as stated in \cref{thm:correctness_complexity} is proved following the same lines as in \cite{karnin2013almost}. We provide it here  for completeness, while adapting it to our problem formulation. First, we need the following result that was proven in \cite{karnin2013almost} for a different type of concentration inequalities.
\begin{lemma} \label{lemma:not_eliminated}
    Assume a set of channels $\roundchannels$ that contains the best channel $\channelidx^\star$. Let now $\pacchannelround \in \roundchannels$ be an $(\epsround/2, \deltaround)$-best arm. Then, for estimations based on $\pullsround$ samples according to \eqref{eq:pullsround}, we have with probability at least $1-2\deltaround$ that
    \begin{align}
        \capacitychannelestr[\channelidx^\star] \geq \capacitychannelestr[\channelidx_\epsci] - \epsround. \label{eq:best_arm_stays}
    \end{align}
\end{lemma}
\begin{proof}
    The proofs follows the same lines as in \cite{karnin2013almost}. We first observe that $\capacitychannel[\pacchannelround] \leq \capacitychannel[\channelidx^\star]$. Hence, from \cref{thm:req_samples} we have that $\capacitychannelestr[\pacchannelround] \leq \capacitychannel[\pacchannelround] + \epsround/2 \leq \capacitychannel[\channelidx^\star] + \epsci/2$ with probability at least $1-\deltaround$. Thus, $\Pr\inpara{\capacitychannelestr[\pacchannelround] \geq \capacitychannel[\channelidx^\star] + \epsround/2} \leq \deltaround$. From \cref{thm:req_samples}, we further have $\Pr\inpara{\capacitychannel[\channelidx^\star] \geq \capacitychannelestr[\channelidx^\star] + \epsround/2} \leq \delta$. Applying the union bound yields
    \begin{align*}
        \Pr\inpara{\capacitychannelestr[\pacchannelround] \geq \capacitychannelestr[\channelidx^\star] + \epsround/2} \leq 2\deltaround.
    \end{align*}
    Taking the complementary event concludes the proof.
\end{proof}

\begin{lemma}[{\cite[Lemma 3.4]{karnin2013almost}}] \label{lemma:decreasing_set}
    Consider \cref{alg:bai}, then $\Pr\inpara{\card{\channelpartround} > \frac{1}{8} \card{\channelpartround[\roundidx-1]}} \leq 24\deltaround$ for any $\roundidx \geq s \geq 1$.
\end{lemma}
\begin{proof}
    The proof of \cite{karnin2013almost} holds verbatim, since our subroutine outputs an $(\epsround/2, \deltaround)$-best channel from either $\pacpolicy$ or \extstart \medianpac. \extend
\end{proof}

We are now ready to prove the correctness of \cref{alg:bai} stated in \cref{thm:correctness_complexity}.
\begin{proof}[Proof of \cref{thm:correctness_complexity}]
Consider \cref{alg:bai} with confidence level $1-\confidence$. 
We utilize \cref{thm:req_samples} with confidence level $\confci=\deltaround=\frac{\confidence}{50\roundidx^3}$ and a maximal deviation of $\epsci=\frac{\epsround}{2} = \frac{2^{-r}}{8}$, to obtain that 
\extstart
\begin{align*}
    \pullsround \geq \frac{4\linfac}{\epsround^2} \log^2\left(\frac{4 \logfac \nfactor[\deltaround]}{\epsround^2}\right).
\end{align*}
\extend
By that choice of $\pullsround$, we have by \cref{lemma:not_eliminated} that \cref{alg:bai} removes a best channel from the set $\roundchannels$ in any of the rounds with probability at most $2\deltaround$. From a union bound over all rounds $\roundidx$ with $\deltaround = \delta/(50\roundidx^2)$, the probability that the best channel gets eliminated in any of the rounds is at most $\sum_{\roundidx}^\infty 2\deltaround = \frac{2\pi^2 \confidence}{6 \cdot 50} \leq \confidence/5$. We apply \cref{lemma:decreasing_set} to prove a decreasing cardinality of $\roundchannels$ as $\roundidx$ grows. By union bound, the probability that the cardinality of $\channelpartround$ decreases by less than a factor of $8$ in all of the rounds $\roundidx>s$ is at most $\sum_{\roundidx=1}^\infty 24\deltaround = \frac{24 \pi^2 \confidence}{6\cdot 50 r^2} \leq 4\confidence/5$. Hence, the algorithm terminates and outputs $\channelidx^\star$ with probability at least $1-\confidence$, which concludes the proof.
\end{proof}

\subsection{Proof of \cref{lemma:tr}} \label{app:proof_lemma_tr}

For notational convenience, we denote the quantities $\tmpx \define 4 \card{\symtx} \card{\symrx}$, $\logfactor \define 50 \card{\symtx}$ and $\rootfactor \define 4 \sqrt[3]{\card{\symtx}} > \sqrt[3]{\logfactor}$. We then further denote $\logrcubed \define \log\inpara{\frac{\logfactor \roundidx^3}{\confidence}}$.
\extstart Throughout the analysis of the sample complexity, we use $\logfac \define \logfacvar[1] = \sqrt[5]{\frac{\card{\symrx}e^2}{\card{\symtx} \log \card{\symtx}}}$ as a constant upper bound for $\logfacvar$, i.e., $\logfac \leq \logfacvar$ for all $\alpha>0$.\extend
\begin{proof}[Proof of \cref{lemma:tr}]
We make use of the following inequality that holds for any $a,b \in \mathbb{R}$:
\begin{align}
    (a+b)^2\leq 2(a^2 + b^2). \label{eq:freshmans_dream}
\end{align}
\extstart
We first bound $\sum_{r=1}^{s-1} \pullsround$ as 
\begin{align}
    &\sum_{r=1}^{s-1} \pullsround = \sum_{r=1}^{s-1} \frac{4 \linfac \nfactor[\deltaround]}{\epsround^2} \log^2\left(\frac{4 \logfac \nfactor[\deltaround]}{\epsround^2}\right) \nonumber \\
    &\stackrel{(a)}{=} \sum_{r=1}^{s-1} 64 \linfac \tmpx \cdot 4^{r} \log\inpara{\frac{\logfactor\roundidx^3}{\confidence}} \log^2\left(64 \logfac \tmpx \cdot 4^{r} \log\inpara{\frac{\logfactor \roundidx^3}{\confidence}}\right) \nonumber \\
    &\stackrel{(b)}{\leq} \sum_{r=1}^{s-1} 64 \linfac \tmpx \cdot 4^{r} \logrcubed \cdot 2\left( r^2 \log^2(4) + \log^2\left(64 \logfac \tmpx \logrcubed\right) \right) \nonumber \\[-0.1cm]
    &\stackrel{(c)}{<} \sum_{r=1}^{s-1} 384 \linfac \tmpx \cdot 4^{r} \log\inpara{\frac{\rootfactor \roundidx}{\confidence}} \Bigg( r^2 \log^2(4) \nonumber \\[-0.3cm]
    &\hspace{3.5cm} + \log^2\left(192 \logfac \tmpx \log\inpara{\frac{\rootfactor \roundidx}{\confidence}} \right) \Bigg), \nonumber \\[-0.7cm] \nonumber
\end{align}
\extend
where $(a)$ is obtained by plugging the values of $\nfactor[\deltaround]$ and $\epsround$; $(b)$ is obtained by using \eqref{eq:freshmans_dream}, and $(c)$ follows from $\logrcubed = \log\inpara{\frac{\logfactor\roundidx^3}{\confidence}} \leq \log\inpara{\frac{\rootfactor \roundidx}{\confidence}}^3 = 3\log\inpara{\frac{\rootfactor \roundidx}{\confidence}}$ and yields \cref{lemma:tr}.
\end{proof}

\subsection{Proof of \cref{lemma:trs}} \label{app:proof_lemma_trs}

To prove the result of \cref{lemma:trs}, we need the definitions from \cref{app:proof_lemma_tr}, the results of \cref{lemma:not_eliminated,lemma:decreasing_set} and the following intermediate results.
\begin{claim} \label{prop:infsums}
    We have the following upper bounds
    \begin{align*}
        \sum_{r=0}^\infty \frac{r^2}{2^r} = 6 \text{, } \extstart \sum_{r=1}^\infty \frac{\log(r)}{2^r} \leq 1 \text{ and } \extend \sum_{r=1}^\infty \frac{r^2 \log(r)}{2^r} \leq 8.
    \end{align*}
    
\end{claim}

\begin{proposition} \label{prop:log_bound}
    Let $\roundidx, s$ be numbers in $\mathbb{N}^+$, $\rootfactor, \logfactor >0$ such that $\rootfactor > \sqrt[3]{\logfactor}$ and $0\leq \confidence\leq 1$, then
    \begin{align*}
    \logrcubed[\roundidx+s] = \log\inpara{\frac{\logfactor(\roundidx+s)^3}{\confidence}} &\leq 3 \log\left(\frac{\rootfactorsplit s}{\delta}\right) + 3 \log\left(\roundidx \right).
\end{align*}
\begin{proof} %
We bound the given quantity as
\begin{align*}
    \log\inpara{\frac{\logfactor(\roundidx+s)^3}{\confidence}} &\leq 3 \log\left(\frac{\rootfactor (\roundidx+s)}{\delta}\right) \nonumber   \leq 3 \log\left(\frac{2 \cdot \rootfactor \roundidx s}{\delta}\right),
\end{align*}
and split the terms to conclude the proof.    
\end{proof}
\end{proposition}

We are now ready to prove \cref{lemma:trs}.
\begin{proof}[Proof of \cref{lemma:trs}]
To bound $\sum_{r=0}^\infty (\frac{1}{8})^{r+1} \pullsround[\roundidx+s]$, we first split the sum into $\frac{1}{8} \pullsround[s]$ and $\sum_{r=1}^\infty (\frac{1}{8})^{r+1} \pullsround[\roundidx+s]$. We continue to bound the latter and will later bound the former.
\extstart
\begin{align}
    \sum_{r=1}^\infty \left(\frac{1}{8}\right)^{r+1} \pullsround[\roundidx+s] &= \sum_{r=1}^\infty \inpara{\frac{1}{8}}^{r+1} \frac{4 \linfac \tmpx \log(\card{\symtx}/\delta_{\roundidx+s})}{\varepsilon_{\roundidx+s}^2} \log^2\left(\frac{4 \logfac \tmpx \log(\card{\symtx}/\delta_{\roundidx+s})}{\varepsilon_{\roundidx+s}^2}\right) \nonumber \\
    &\stackrel{(a)}{=} \sum_{r=1}^\infty \inpara{\frac{1}{8}}^{r+1} 64 \linfac \tmpx \cdot 4^{r+s} \logrcubed[\roundidx+s] \log^2\left(4^{\roundidx+s} \cdot 64 \logfac \tmpx \logrcubed[\roundidx+s]\right) \nonumber \\
    &\stackrel{(b)}{\leq} \sum_{r=1}^\infty \frac{64 \linfac \tmpx 2^{-r} 4^s}{8} \logrcubed[\roundidx+s] 2 \left( \log^2\!\left(4^{r+s} \right) \! + \! \log^2\left(64 \logfac \tmpx \logrcubed[\roundidx+s]\right) \right) \nonumber \\
    &= \frac{4^s \cdot 64 \linfac \tmpx}{4} \sum_{r=1}^\infty  2^{-r} \logrcubed[\roundidx+s] \Big( \underbrace{(r+s)^2 \log^2\left(4 \right)}_{\text{term 1}}
    +\underbrace{\log^2\left(64 \logfac \tmpx \logrcubed[\roundidx+s]\right)}_{\text{term 2}} \Big), \label{intermediate_second_complexity}
\end{align}
where $(a)$ holds because $\frac{4^\roundidx}{8^{\roundidx+1}} = \frac{2^{2\roundidx}}{2^{3{\roundidx}} \cdot 8} = \frac{1}{8 \cdot 2^\roundidx}$\extend, and $(b)$ follows from \eqref{eq:freshmans_dream}.
\extstart Next, for ease of notation, we define $\tmpcrs \define 96 \linfac \tmpx$. \extend Then, term $1$ of \eqref{intermediate_second_complexity} is bounded using \cref{prop:log_bound} and \eqref{eq:freshmans_dream} as
\extstart
\begin{align}
    &\frac{4^s}{4} \cdot 64 \linfac \tmpx \sum_{r=1}^\infty  2^{-r} \cdot  \logrcubed[\roundidx+s] \left( (r+s)^2 \log^2\left(4 \right) \right) \nonumber \\
    &\leq \tmpcrs 4^s \left( \log\left(\frac{\rootfactorsplit s}{\delta}\right) \sum_{r=1}^\infty  \frac{(r+s)^2}{2^{r}} + \sum_{r=1}^\infty  \frac{\log\left(r\right) (r+s)^2}{2^{r}} \right) \nonumber \\
    &\stackrel{(a)}{\leq} 2 \tmpcrs 4^s \left( \log\left(\frac{\rootfactorsplit s}{\delta}\right) \sum_{r=1}^\infty  \left(\frac{r^2}{2^{r}} + \frac{s^2}{2^{r}}\right)\right. \left. \sum_{r=1}^\infty \left( \frac{\log\left(r\right) r^2}{2^{r}} + \frac{\log\left(r\right) s^2}{2^{r}}\right) \right)  \nonumber \\
    &\stackrel{(b)}{=} 2 \tmpcrs 4^s \left( \log\left(\frac{\rootfactorsplit s}{\delta}\right) \left(6 + s^2\right) + \left( 8 + s^2\right) \right) \nonumber \\
    &= \mathcal{O} \left(s^2 4^s \log\left(\frac{s}{\delta}\right)\right), \label{eq:first_term}
\end{align}
\extend
where $(a)$ follows from \eqref{eq:freshmans_dream} and $(b)$ follows using that $\sum_{x=r}^{\infty} \frac{1}{2^r} = 1$ and by applying the results of \cref{prop:infsums}.

To bound term $2$ in \eqref{intermediate_second_complexity}, we first establish two bounds. To this end, we will need the following intermediate result.

\begin{proposition} \label{prop:bound_logsq_rs} Let $c > 0$ be a constant and $\roundidx, s \in \mathbb{N}^+$, then we have for some $\rootfactor > \sqrt[3]{\logfactor}$ the following inequality:
    \begin{align*}
        \log^2\inpara{c \logrcubed[\roundidx+s]} &= \log^2 \left(c \log\inpara{\frac{\logfactor (\roundidx+s)^3}{\confidence}}\right) \\
        &< 2 \log^2 \left(3c \log\left(\frac{\rootfactorsplit s}{\delta}\right)\right) + 2 \roundidx^2
    \end{align*}
\end{proposition}
\begin{proof} %
We use \cref{prop:log_bound} and the fact that $\log(a+b) = \log(a \cdot (1+b/a))$ and $\log(a)\leq a-1$. Hence, we obtain
\begin{align}
    &\log^2 \left(c \log\inpara{\frac{\logfactor (\roundidx+s)^3}{\confidence}}\right) \leq \log^2 \left(3c \log\left(\frac{\rootfactorsplit rs}{\delta}\right)\right) \nonumber \\
    &= \log^2 \left(3c \left( \log\left(\frac{\rootfactorsplit s}{\delta}\right) + \log\left(r\right)\right)\right) \nonumber \\
    &= \log^2 \left(3c \log\left(\frac{\rootfactorsplit s}{\delta}\right) \cdot \left(1+\log\left(\frac{r\delta}{\rootfactorsplit s}\right)\right)\right) \nonumber \\
    &\leq \log^2 \left(3c \log\left(\frac{\rootfactorsplit s}{\delta}\right) \cdot \left(1+\log\left(r\right)\right)\right) \nonumber \\
    &\leq \log^2 \left(3c \log\left(\frac{\rootfactorsplit s}{\delta}\right) \cdot r\right) \nonumber \\
    &\leq 2 \log^2 \left(3c \log\left(\frac{\rootfactorsplit s}{\delta}\right)\right) + 2 \log^2\left( r\right) \nonumber \\
    &< 2 \log^2 \left(3c \log\left(\frac{\rootfactorsplit s}{\delta}\right)\right) + 2 (r-1)^2,
\end{align}
and observing that $\roundidx-1 < \roundidx$ and $r\geq1$ concludes the proof. 
\end{proof}

We apply \cref{prop:bound_logsq_rs} with $c=64 \logfac \tmpx$ together with results form \cref{prop:infsums} to derive the relation 
\begin{align}
    &\log\left(\frac{\rootfactorsplit s}{\delta}\right) \sum_{r=1}^\infty  \frac{1}{2^{r}} \log^2\inpara{c \log\inpara{\frac{\logfactor (r+s)^3}{\delta}}} \nonumber \\
    &\leq \log\left(\frac{\rootfactorsplit s}{\delta}\right) \sum_{r=1}^\infty  \frac{1}{2^{r}} \left(2 \log^2 \left(3c \log\left(\frac{\rootfactorsplit s}{\delta}\right)\right) + 2 r^2\right) \nonumber \\
    &= \log\left(\frac{\rootfactorsplit s}{\delta}\right) \left(2 \log^2 \left(3c \log\left(\frac{\rootfactorsplit s}{\delta}\right)\right) \sum_{r=1}^\infty  \frac{1}{2^{r}}  + 2 \sum_{r=1}^\infty  \frac{r^2}{2^{r}} \right) \nonumber \\
    &\leq \log\left(\frac{\rootfactorsplit s}{\delta}\right) \left(2 \log^2 \left(3c \log\left(\frac{\rootfactorsplit s}{\delta}\right)\right) + 12 \right). \label{eq:bound_first_sum}
\end{align}

Again using the result from \cref{prop:bound_logsq_rs} with $c=64 \logfac \tmpx$ together with results from \cref{prop:infsums}, we have that
\begin{align}
    &\sum_{r=1}^\infty  \frac{\log(r)}{2^{r}} \log^2\left(c \log\inpara{\frac{\logfactor (r+s)^3}{\delta}}\right) \nonumber \\
    &\leq \sum_{r=1}^\infty  \frac{\log(r)}{2^{r}} \left(2 \log^2 \left(3c \log\left(\frac{\rootfactorsplit s}{\delta}\right)\right) + 2 r^2\right) \nonumber \\
    &= 2 \log^2 \left(3c \log\left(\frac{\rootfactorsplit s}{\delta}\right)\right) \sum_{r=1}^\infty  \frac{\log(r)}{2^{r}}  + 2 \sum_{r=1}^\infty  \frac{r^2 \log(r)}{2^{r}} \nonumber \\
    &\leq 2 \log^2 \left(3c \log\left(\frac{\rootfactorsplit s}{\delta}\right)\right) + 16. \label{eq:bound_second_sum}
\end{align}

With $c=64 \logfac \tmpx$ and using the results from \eqref{eq:bound_first_sum}, \eqref{eq:bound_second_sum} and \cref{prop:log_bound}, we obtain for the second term in \eqref{intermediate_second_complexity} that
\extstart
\begin{align}
    &\frac{4^s \cdot 64\linfac \tmpx}{4} \sum_{r=1}^\infty  2^{-r} \cdot  \logrcubed[\roundidx+s] \log^2\left(64 \logfac \tmpx \log\inpara{\frac{\logfactor (r+s)^3}{\delta}}\right) \nonumber \\
    &= \frac{4^s 192 \linfac \tmpx}{4} \left(\! \log\left(\!\frac{\rootfactorsplit s}{\delta}\!\right) \sum_{r=1}^\infty  \frac{1}{2^{r}} \log^2\left(\! 64 \logfac \tmpx \log\inpara{\frac{\logfactor (r+s)^3}{\delta}}\!\right) \right. \nonumber \\
    &\left.+ \sum_{r=1}^\infty  \frac{\log\left(r\right)}{2^{r}} \cdot \log^2\left(64 \logfac \tmpx \log\inpara{\frac{\logfactor (r+s)^3}{\delta}}\right) \right) \nonumber \\
    &\leq \begin{aligned}[t] 4^s \cdot 48 \tmpx \cdot &\left( \log\left(\frac{\rootfactorsplit s}{\delta}\right) \left(2 \log^2 \left(3c \log\left(\frac{\rootfactorsplit s}{\delta}\right)\right) + 12 \right) \right. \nonumber \\
    &\left.+ 2 \log^2 \left(3c \log\left(\frac{\rootfactorsplit s}{\delta}\right)\right) + 16 \right) \end{aligned} \\
    &= \mathcal{O} \left( 4^s \log\left(\frac{s}{\delta}\right)  \log^2\log\left(\frac{s}{\delta}\right) \right). \label{eq:second_term}
\end{align}
\extend
To conclude the proof, we bound $\frac{1}{8} \pullsround[s] \leq \pullsround[s]$ as 
\extstart
\begin{align}
    &\pullsround[s] = 64 \linfac \tmpx \cdot 4^{s} \log\inpara{\frac{\logfactor s^3}{\confidence}} \log^2\left(64 \logfac \tmpx \cdot 4^{r} \log\inpara{\frac{\logfactor s^3}{\confidence}}\right) \nonumber \\
    &\leq 64 \linfac \tmpx \cdot 4^{s} \logrcubed[s] \cdot 2\left( s^2 \log^2(4) + \log^2\left(64 \logfac \tmpx \logrcubed[s] \right) \right) \nonumber \\[-0.1cm]
    &< 384 \tmpx \cdot 4^{s} \log\inpara{\frac{\rootfactor s}{\confidence}} \Bigg( s^2 \log^2(4) \nonumber \\[-0.3cm]
    &\hspace{3.5cm} + \log^2\left(192 \logfac \tmpx \log\inpara{\frac{\rootfactor s}{\confidence}} \right) \Bigg) \nonumber \\[-0.1cm]
    &= \mathcal{O} \left( s^2 4^s \log\left(\frac{s}{\delta}\right) + 4^s \log\left(\frac{s}{\delta}\right) \log^2\log\left(\frac{s}{\delta}\right) \right). \label{eq:r0}
\end{align}
\extend
Putting \eqref{eq:first_term}, \eqref{eq:second_term} and \eqref{eq:r0} together, we obtain the result stated in \cref{lemma:trs}. %
\end{proof}

\subsection{Proof of \cref{lemma:trinf}} \label{app:proof_lemma_trinf}

\begin{proof}
    
We choose $\pullsround$ to fulfill \eqref{eq:req_samples} with equality. Given that, we bound $\sum_{r=1}^{\infty} \inpara{\frac{1}{8}}^{r-1} \pullsround$ from above as 
\extstart
\begin{align}
    &\sum_{r=1}^{\infty} \inpara{\frac{1}{8}}^{r-1} \!\!\! \pullsround = \sum_{r=1}^{\infty} \inpara{\frac{1}{8}}^{r-1} \frac{4 \linfac \nfactor[\deltaround]}{\epsround^2} \log^2\left(\frac{4 \logfac \nfactor[\deltaround]}{\epsround^2}\right) \nonumber \\
    &\stackrel{(a)}{=} \sum_{r=1}^{\infty} \inpara{\frac{1}{8}}^{r-1} \!\!\!\!\! 64 \linfac \tmpx \cdot 4^{r} \log\inpara{\frac{\logfactor\roundidx^3}{\confidence}} \log^2\left(64 \logfac \tmpx \cdot 4^{r} \log\inpara{\frac{\logfactor \roundidx^3}{\confidence}}\right) \nonumber \\
    &\stackrel{(b)}{\leq} \sum_{r=1}^{\infty} \frac{8}{2^\roundidx} \cdot 64 \linfac \tmpx \logrcubed \cdot 2\left( r^2 \log^2(4) + \log^2\left(64 \logfac \tmpx \logrcubed\right) \right) \nonumber \\[-0.1cm]
    &\stackrel{(c)}{<} \sum_{r=1}^{\infty} \frac{8}{2^\roundidx} \cdot 384 \linfac \tmpx  \log\inpara{\frac{\rootfactor \roundidx}{\confidence}} \Bigg( r^2 \log^2(4) \nonumber \\[-0.3cm]
    &\hspace{3.5cm} + \log^2\left(192 \logfac \tmpx \log\inpara{\frac{\rootfactor \roundidx}{\confidence}} \right) \Bigg), \nonumber \\[-0.7cm] \nonumber
    \\
    &\stackrel{(d)}{<} \sum_{r=1}^{\infty} \frac{8}{2^\roundidx} \cdot 384 \linfac \tmpx  \log\inpara{\frac{\rootfactor \roundidx}{\confidence}} \Bigg( r^2 \log^2(4) \nonumber \\[-0.3cm]
    &\hspace{3.5cm} + 2 \log^2 \left(192 \logfac \tmpx \log\left(\frac{\rootfactor}{\delta}\right)\right) + 2 \roundidx^2 \Bigg) \nonumber \\[-0.7cm] \nonumber
    \\
    &\stackrel{(e)}{=} \mathcal{O} \left( \log\left(\frac{1}{\delta}\right)  \log^2\log\left(\frac{1}{\delta}\right)\right), \nonumber
\end{align}
\extend
where $(a)$ is obtained by plugging the values of $\nfactor[\deltaround]$ and $\epsround$, $(b)$ is obtained by using that \extstart $\frac{4^\roundidx}{8^{\roundidx-1}} = \frac{8 \cdot 2^{2\roundidx}}{2^{3{\roundidx}}} = \frac{8}{2^\roundidx}$ \extend and \eqref{eq:freshmans_dream}, and $(c)$ follows from $\logrcubed = \log\inpara{\frac{\logfactor\roundidx^3}{\confidence}} \leq \log\inpara{\frac{\rootfactor \roundidx}{\confidence}}^3 = 3\log\inpara{\frac{\rootfactor \roundidx}{\confidence}}$, $(d)$ is obtained by the same arguments as the proof of \cref{prop:bound_logsq_rs}, and $(e)$ follows from \cref{prop:infsums}. This concludes the proof.
\end{proof}

\extstart

\subsection{Proof of \cref{thm:correctness-complexity-pac-median}}
\label{sec:app_thm3}
\begin{proof}%

The proof of correctness follows similar lines to the proof in \cite{evendar2006action}, and holds since our sampling strategy guarantees the required target probabilities and worst-case estimation errors. This suffices to assure the correctness of the algorithm. 
To prove the sample complexity stated in \cref{thm:correctness-complexity-pac-median}, we require the simplification of the following function $\logsq$.
\begin{align*}
    &\logsq = \log^2\left(\frac{4 \logfac \nfactor[\deltaround/3]}{\epsround^2}\right) = \log^2\left(\frac{4 \logfac  \nfactor[\delta/3/2^\roundidx]}{\left((3/4)^{\roundidx-1} \epsci/4\right)^2}\right) \\
    &= \left(\log\left( (16/9)^{\roundidx-1}\right) + \log\left( 16 \cdot 4 \logfac \nfactor[\delta/3/2^\roundidx]/\epsci^2 \right)\right)^2 \\
    &\leq 2 \left( (\roundidx-1)^2 \log^2\left(16/9\right) + \log^2\left( 16 \cdot 4 \logfac \nfactor[\delta/3/2^\roundidx] / \epsci^2 \right) \right) \\
    &\leq C_5 (\roundidx-1)^2 + 2 \log^2\left( \frac{256 \logfac \card{\symtx} \card{\symrx}}{\epsci^2} \log(2^\roundidx 3\card{\symtx}/\delta) \right) \\
    &\leq C_5 (\roundidx-1)^2 + 2\log^2\left( \frac{C_3}{\epsci^2} \left(\log(2^\roundidx) + \log(3\card{\symtx}) + \log\left(\frac{1}{\delta}\right) \right) \right) \\
    &\leq C_5 (\roundidx-1)^2 + 2 \log^2\left( \frac{C_3}{\epsci^2} \log\left(\frac{1}{\delta}\right) \left(\roundidx C_1^\prime + C_1 \right) \right) \\
    &\leq C_5 (\roundidx-1)^2 + 4\log^2\left( \frac{C_3}{\epsci^2} \log\left(\frac{1}{\delta}\right) \right) + 4\log^2\left( \roundidx C_1^\prime + C_1 \right) \\
    &\leq C_2 (\roundidx-1)^2 + 4\log^2\left( \frac{C_3}{\epsci^2} \log\left(\frac{1}{\delta}\right) \right) + 8\log^2\left( \roundidx C_1^\prime\right) + C_4 \\
    &\leq C_2 (\roundidx-1)^2 + 4\log^2\left( \frac{C_3}{\epsci^2} \log\left(\frac{1}{\delta}\right) \right) + 16 ( \roundidx-1 )^2 + C_6 \\ %
    &\leq C_5 (\roundidx-1)^2 + 4\log^2\left( \frac{C_3}{\epsci^2} \log\left(\frac{1}{\delta}\right) \right) + C_6,
\end{align*}
where $C_2 = 2\log^2\left(16/9\right)$, $C_3 = 256 \logfac \card{\symtx} \card{\symrx}$, $C_4 = 8\log^2(C_1)$, $C_5 = 16 + C_2$, $C_6 = 16 \log^2(C_1^\prime) + C_4$.
Now assuming that $\epsci \leq 4$ such that $\forall \roundidx: \epsround \leq 1$, we have
\begin{align*}
    &\sum_{\roundidx=1}^{\log_2(\nchannels)} \frac{k}{2^{\roundidx-1}} \frac{4 \linfac \nfactor[\deltaround/3]}{\epsround^2} \cdot \logsq \\
    &= \sum_{\roundidx=1}^{\log_2(\nchannels)} \frac{k}{2^{\roundidx-1}} \frac{4 \linfac \nfactor[\delta/3/2^\roundidx]}{\left((3/4)^{\roundidx-1} \epsci/4\right)^2} \cdot  \logsq \\
    &= \frac{64 \linfac \nchannels}{\epsci^2} \sum_{\roundidx=1}^{\log_2(\nchannels)} \left(\frac{8}{9}\right)^{\roundidx-1} \nfactor[\delta/3/2^\roundidx] \cdot  \logsq \\
    &= \frac{64 \linfac \nchannels}{\epsci^2} \sum_{\roundidx=1}^{\log_2(\nchannels)} \left(\frac{8}{9}\right)^{\roundidx-1} 4 \card{\symtx} \card{\symrx} \log(2^\roundidx3\card{\symtx}/\delta) \cdot \logsq \\
    &= \frac{C \cdot \nchannels}{\epsci^2} \sum_{\roundidx=1}^{\log_2(\nchannels)} \left(\frac{8}{9}\right)^{\roundidx-1} \left(\roundidx \log(2) + \log(3\card{\symtx}) + \log\left(\frac{1}{\delta}\right)\right) \logsq \\
    &\leq \begin{aligned}[t] \frac{C \cdot \nchannels}{\epsci^2} \sum_{\roundidx=1}^{\log_2(\nchannels)} &\left(\frac{8}{9}\right)^{\roundidx-1} \left(\roundidx C_1^\prime + C_1 + \log\left(\frac{1}{\delta}\right)\right) \cdot \\
    &\left(C_5 (\roundidx-1)^2 + 4\log^2\left( \frac{C_3}{\epsci^2} \log\left(\frac{1}{\delta}\right) \right) + C_6\right) \end{aligned} \\
    &= \mathcal{O}\left( \frac{\nchannels}{\epsci^2} \log\left(\frac{1}{\delta}\right) \sum_{\roundidx=1}^{\infty} \left(\frac{8}{9}\right)^{\roundidx-1} \left( \roundidx^3 + r \log^2\left( \frac{1}{\epsci^2} \log\left(\frac{1}{\delta}\right) \right)  \right) \right) \\
    &= \mathcal{O}\left( \frac{\nchannels}{\epsci^2} \log\left(\frac{1}{\delta}\right) \left( \log^2\left( \frac{1}{\epsci^2}\log\left(\frac{1}{\delta}\right) \right) \right) \right),
\end{align*}
where $C=256 \linfac \card{\symtx} \card{\symrx}$, $C_1 = \log(3\card{\symtx})$ and $C_1^\prime = \log(2)$. This concludes the proof.
\end{proof}
\extend

\bibliographystyle{IEEEtran}
\bibliography{references}

\end{document}